\documentclass[10pt,letter,journal]{IEEEtran}
\usepackage{subcaption}
\usepackage{cite}
\usepackage{amsmath,amsthm,amssymb,amsfonts}
\usepackage{lipsum}
\usepackage{bbm}

\usepackage{mathtools}
\usepackage{bm}
\usepackage{algorithmic}
\usepackage{graphicx}
\usepackage[ruled,vlined]{algorithm2e}
\usepackage{textcomp}
\usepackage{tikz}
\usepackage{xcolor}
\usepackage{svg}
\usepackage{flushend}
\usepackage{balance}
\usepackage[spaces,hyphens]{url}
\renewcommand{\baselinestretch}{0.97}
\theoremstyle{plain}
\newtheorem{lemma}{Lemma}
\newtheorem{corollary}{Corollary}
\newtheorem{theorem}{Theorem}
\theoremstyle{definition}
\newtheorem{definition}{Definition}
\newtheorem{assum}{Assumption}
\newtheorem{remark}{Remark}
\newtheorem{proposition}{Proposition}
\newcommand{\remove}[1]{}

\usepackage{booktabs}

\begin{document}
\title{Whittle Index Policy for Three State Collapsing Bandits}
\author{Rahul Meshram,  Anu Krishna and Kesav Ram Kaza
\thanks{The work of Rahul Meshram is supported through IITM NFIG Grant, IITM NFSG Grant, and ANRF grant Project No EEQ/2021/000812. 
}
\thanks{Rahul Meshram and Anu Krishna is with the Department of Electrical Engineering , IIT Madras, Chennai, India. (e-mail: rahulmeshram@ee.iitm.ac.in, krishna.anu88@gmail.com) }
\thanks{Kesav Ram Kaza is 
now with the Department of Information Technology, IIIT Allahabad, Prayagraj, India. (e-mail: krkaza@iiita.ac.in).}}

\maketitle

\begin{abstract}
Partially observable restless multi-armed bandit models have been studied for dynamic resource allocation problems under uncertainty. This paper analyzes three-state restless multi-armed bandits with partial observability. These are called collapsing bandit problems if the state of an arm (bandit) is perfectly observable when it is scheduled for allocation. The state is not observable when the arm is not scheduled. The planner has a budget constraint on the scheduling decisions in each time slot. The planner's goal is to maximize expected cumulative long term discounted reward. We study the Whittle index policy. To analyze index policy, first  the $N$-armed collapsing bandit problem is decoupled into $N$ single arm problems, and convexity and monotonicity of the optimal value function in belief vector are derived using monotone likelihood ratio (MLR) ordering on state transition dynamics and monotone reward structure assumption.

We next show that the  optimal policy is of threshold type in belief under sufficiency conditions on the model. These sufficiency conditions depend on Lipschitz bound on the value function in the belief. 
We further prove the indexability for an arm using Lipschitz bound on subsidy sensitivity with minimal restrictive assumption on the transition probability model and discount factor $\beta.$ Our results generalize to three state and strengthen existing results on threshold policy and indexability with weaker model assumptions.
We provide index computation algorithm and we present comparison of the Whittle-index, myopic policy and round-robin policies in the simulation study.


\end{abstract}


\section{Introduction}

Resource allocation under uncertainty arises frequently in applications operating within dynamic environments. Restless multi-armed bandits offer a sequential decision-making framework that has been studied and applied to resource allocation problems across several domains including opportunistic wireless communication systems \cite{Ahmad09, LiuZhao10}, uncrewed aerial vehicle scheduling \cite{jleny:2008:phdthesis}, machine maintenance planning \cite{dayanik2008index, glazebrook2014stochastic}, oil exploration \cite{brown2013optimal}, recommendation systems \cite{Meshram17a}, public health intervention planning \cite{Mate2021Bandit,nino2026adherence}, wildfire management \cite{kaza2024constrained}, among others. 


RMABs are a class of sequential decision-making problems in which the planner schedules a collection of ``arms'' over time. Each arm's state evolves according to action-dependent Markovian dynamics, and the planner must operate under a budget constraint, typically integer-valued in nature. This constraint renders the problem inherently challenging to solve. RMABs were first introduced by Whittle in \cite{Whittle88}, who proposed an index-based policy that has since come to be known as the Whittle index policy. A key requirement underlying this approach is indexability, a condition that must be satisfied by each individual arm. RMABs involve arms with only two available actions—play or not play—with the state of each arm fully observable to the planner. More generally, RMABs are known to be PSPACE-hard \cite{Papadimitriou99}. Despite being a heuristic, the Whittle index policy has nevertheless been shown to achieve optimality under asymptotic conditions across a variety of settings.

RMABs have found application in public health intervention planning \cite{Mate2020Bandit} and have been validated in field trials \cite{Mate2021banditField}. This line of work is motivated by the observation that intelligently schedule healthcare interventions can improve patient adherence to medication regimens for conditions such as diabetes, hypertension, tuberculosis, cancer, etc. The underlying goal is to preserve patients' health through prevention and early diagnosis, while sustaining adherence to the prescribed treatment protocol. In  \cite{Mate2021banditField}, the authors applied a two-state, two-action model, where the two states denote adherence and non-adherence of an individual, while the two actions denote whether an individual is contacted for counseling/information collection or not. Field trials showed the effectiveness of the RMAB model for scheduling call to members across the target population.

A finite state representation of an individual health with more than two states might allow us to capture different levels of severity of  health/ill health, which is not possible using a two-state model. In view of furthering the applicability of RMABs to healthcare planning, in this paper we analyze a three-state two-action RMAB model.

\subsection{Related Work}
%
We shall briefly review the work on partially observable RMABs (PORMAB). The literature on PORMABs has mostly focused on the analysis of two-state models owing to their tractability relative to multi-state models. In \cite{LiuZhao10}, authors studied two state PORMABs where the state of an arm (bandit) is observable when it is scheduled and not observed when it is not scheduled. They characterised the  structure of the optimal value function and showed the indexability under conditions on state transition dynamics. In \cite{Wang13,Wang14},  the optimality of myopic policy is studied under some strict model structure valid for some opportunistic communication applications. In \cite{Meshram18},  a model studied where uncertainty about an arm's state persists even when the arm is scheduled, although it is lower than when the arm is not scheduled. In \cite{kaza2019whittle}, authors generalized this model by allowing multiple state transitions in each decision interval (time slot). There is an imperfect observation signal, a \textit{cumulative feedback} about the effect of these state transitions at the end of the decision interval. In \cite{dayanik2008index} and  \cite{kaza2024constrained}, these work considered another generalization of the model where arms are constrained with stochastic availability. In these models \cite{Meshram18,kaza2019whittle}, the planner has perennial uncertainty (fluctuates based on the decision sequence but always exists) about the state. In  \cite{Mate2020Bandit} studied a two-state model in the context of a health care application and renamed these bandits as \textit{collapsing bandits} as the uncertainty about state collapsed when an arm is scheduled. Following \cite{Mate2020Bandit}, we adopt the term \textit{collapsing bandits}, motivated in part by their applications in healthcare.


In the work of \cite{LiuZhao10}, model is collapsing bandit with state perfectly observable for active action, otherwise no state observation, they studied for two state model. Here, transition probabilities for both active and passive action are same and this assumption simplifies the analysis of threshold type policy, indexability and explicit index formula is derived. 
Later in \cite{Mate2020Bandit},   collapsing bandit model with two state is studied with different transition probability matrices for active and passive action. It increases analytical difficulty, the threshold type policy and indexability is studied with  assumption on transition probabilities and discount parameter $\beta.$ Their results do not hold true for $\beta \rightarrow 1.$ They  need stronger conditions on transition probabilities and discount factor $\beta.$
In the work of \cite{Meshram18}, authors studied two state model with state never observable for active action, and passive action. The stronger conditions on transition probabilities are introduced for threshold type policy and further indexability requires a discount factor $\beta <1/3.$ This model is a generalization of both models in  \cite{LiuZhao10} and \cite{Mate2020Bandit}. These work uses value function based Whittle index and indexability approach.
In \cite{nino2026adherence}, authors  considered two-state collapsing bandits for the treatment adherence outreach problem under capacity constraints. They used partial conservation laws and studied PCL indexability and provide PCL conditions for indexability and threshold structure.

Recently, work of \cite{Liu2025relaxed} studied multi state collapsing bandit model, it is an extension of \cite{LiuZhao10} to finite state (more than two state) model. The indexability is proved with restriction on discount factor $\beta < 0.5.$ The index formula is derived and relaxed indexability is introduced. 
In \cite{Akbarzadeh2022conditions} studied a PO-RMAB (finite state and two actions) model with a \textit{controlled restart} assumption,  In these models, they require stronger condition on  transition probabilities and discount factor $\beta$ for indexability. 
Based on this, the  indexability is proved and index computation algorithm is developed. 

All these work \cite{Akbarzadeh2022conditions, LiuZhao10, Liu2025relaxed} suggest that there is trade-off between model assumptions on transition matrices and the discount factor for indexability results in partially observable RMAB. Further, for more than two state models, we require stronger conditions on model or $\beta$ for indexability.

In this paper, we study indexability for collapsing bandit with three state having TP2 transition matrices, we provide indexability result and demonstrate trade-off between transition probability model assumption and discount factor $\beta.$ We further allow $\beta \rightarrow 1.$ Our results therefore broaden the range of models and discount factors for which indexability can be established.


%
\subsubsection{Our Contributions}
The main contributions of this paper are as follows. 
\begin{enumerate}
    \item We formulate the problem of collapsing restless multi-armed bandits for three state model and it is partially observable RMAB. In the collapsing bandit model,  active action reveals the state while while passive action do not reveal state. 
    
    \remove{\item We decouple the joint RMAB problem into $N$ single armed restless bandit using Lagrangian relaxation method.}
    
    \item  We provide the structure results for a single-armed restless bandit-convexity of value functions for passive action, active action, an optimal value function. Later, monotonicity of optimal value function is proved under assumption on monotonicity and TP2 on transition matrices and Lipschtiz bound is derived for value function. 
    \item Under MLR ordering conditions on the
    state-transition dynamics, a suitable monotonicity condition
    on the reward structure, and a Lipschtiz bound, we establish sufficient conditions under
    which an optimal policy has a threshold-type structure over the
    belief space. These results characterize the structure of the
    optimal action regions in the three-state setting. 
    \item We provide a Lipschitz bound on subsidy sensitivity and using it, we show sufficiency condition for indexability of a single-armed restless bandit.  The sufficient
    condition is expressed in terms of the discount factor and the
    discrepancy between the beliefs induced by activation and those
    predicted under passive evolution. This provides a verifiable
    criterion for indexability in the three-state collapsing-bandit
    model. We also establish indexability for the birth-dealth model and the restart model.
    
    \item The Lipschtiz bound in subsidy needs the differentiability of value function in subsidy and we provide results with differentiability and separately for non-differentiability points. In derivation of these bounds, we extensively use property of convexity of value function in subsidy, use Hausdorff distance  for comparison at the point of non-differentiability of value function, left and right derivative. This is one of the important result in our paper. We further provide insights on  discount factor approaches to $1,$ and the condition on the model.    
    
    \item We present an index computation scheme based on two-timescale stochastic approximation algorithm and we also discuss bisection method for index computation. Both algorithms require value function to run on faster timescale. Index computation is faster in bisection as it uses successive interval halving.
    \item 
    We evaluate the performance of the Whittle-index policy against myopic (greedy) and round-robin policies via numerical simulations with homogeneous and heterogeneous arms.  We observe that the Whittle index policy performs better than myopic and round-robin policies. 
\end{enumerate}

Our novel contribution in paper is that we established threshold type policy in belief and proving indexability using Lipschitz bound on value function in belief and Lipsichtz bound on subsidy sensitivity with minimal restrictive assumption on the transition probability model and discount factor $\beta.$ 

The paper is organized as follows. In section \ref{sec:model}, we formulate collapsing bandit model and obtain the optimal dynamic program. Section \ref{sec:decoupled-RMAB} studies decoupled RMAB problem. The value function properties and optimal policy of threshold type for a single-armed restless bandit is developed in section~\ref{sec:SAB}. Indexability and Whittle index policy is studied in section \ref{sec:Whittle-index}. Numerical examples are presented in section~\ref{sec:sim-example} and concluding remark in section~\ref{sec:conclusion}.


\section{Mathematical Model}
\label{sec:model}

Consider a \textit{partially observable $N$-armed  restless bandits}, where arms are assumed to be independent. Each arm, say $n$, $1 \leq n \leq N$, is a partially observable Markov decision process (POMDP) and it is denoted by  
$\mathcal{M}_n = \{\mathcal{S}_n, \mathcal{A}_n, \mathcal{P}_n, \mathcal{R}_n, \mathcal{O}_n, \mathcal{Z}_n, \beta\}.$
We assume that each arm has three states and there are two actions, i.e., active $1$ and passive $0.$ Hence, the state space is $\mathcal{S}_n = \{0,1,2\},$
and action space 
$\mathcal{A}_n = \{0,1\}.$
For arm $n$, $\mathcal{P}_n^{a} = \big[p^a_{n,ij}\big],$ represents the transition probability matrix  under action $a,$  where $p^a_{n,ij}$ denotes the probability of transitioning from state $i$ to state $j$ upon taking action $a \in \mathcal{A}_n$.

Time proceeds in discrete steps, it is  denoted as $t$. An arm $n$ occupies state $s_n(t) \in \mathcal{S}_n$, and the planner chooses an action $a_n(t) \in \mathcal{A}_n$ at time $t$. This yields a reward $r(s_n(t),a_n(t)),$  where $\mathcal{R}_n$ is the associated reward matrix. $\beta$ is the discount factor and $0 <\beta <1.$

The planner does not observe the state $s_n(t)$ for arm $n$ if the action is passive, i.e., $a_n(t) = 0.$ While the planner can observe the state $s_n(t)$ if the action is active, i.e., $a_n(t) = 1.$
Hence the observations from arm $n$ is given by
$\mathcal{O}_n = \{0,1, 2\}$ after taking an action $a_n(t)=1.$ 


When the arm is passive, the current state is not observed. Hence the planner keeps track of a \textit{belief} vector which is probability distribution over the state space. This belief is updated based on the observation, according to Bayes rule. Arm $n$'s belief at time $t$ is defined as
$\omega_n^s(t)=\Pr(s_n(t)=s \mid H(t),\bm{\omega}_n(0)),$
starting from the initial belief $\bm{\omega}_n(0)=[\omega_n^0(0),\omega_n^1(0),\omega_n^2(0)]^T$. Since it is a probability mass function, belief $\omega_n:= [\omega_n^s]_{s\in\mathcal{S}_n}$ obeys $\lVert \omega_n(t) \rVert_1 = \sum_{s} \bm{\omega}_n^s(t) =1.$
For ease of reference, the main notation used throughout the paper is summarized in Table~\ref{tab:notation}.
\begin{table*}[t]
\caption{Notation used in the paper}
\label{tab:notation}
\centering
\renewcommand{\arraystretch}{1.15}
\begin{tabular}{ll}
\hline
\textbf{Symbol} & \textbf{Description} \\
\hline
$N$ & Number of arms in the restless multi-armed bandit system. \\

$\mathcal S_n$ & State space of arm $n$, with
$\mathcal S_n=\{0,1,2\}$. \\

$\mathcal A$ & Admissible set of action space \\

$\mathcal A_n$ & Set of action from $n$\\

$\phi$ & Policy\\
$\Phi$ & Set of policies\\
$P_n^a$ & State-transition matrix of arm $n$ under action $a$. \\

$\beta$ & Discount factor, with $0<\beta<1$. \\

$B$ & Maximum number of arms that can be activated in each time slot. \\

$o_n(t)$ & Observation obtained from arm $n$ at time $t$. \\

$\mathcal O_n$ & Set of observations available from arm $n$. \\

$\bm{\omega}_n(t)$ & Belief vector over the states of arm $n$ at time $t$. \\

$\Pi$ & Belief simplex for a three-state arm,
$\Pi=\{\bm{\omega}\in\mathbb R^3:
\omega_i\geq 0,\ \sum_{i=0}^{2}\omega_i=1\}$. \\

$\Omega$ & Initial joint belief state of all arms,
$\Omega=(\bm{\omega}_1,\ldots,\bm{\omega}_N)$. \\

$\tau_n(\bm{\omega},o,a)$ & Updated belief of arm $n$ after action $a$
and observation $o$. \\

$\Gamma_{n,i}$ & Posterior belief following activation and observation
of state $i$ for arm $n$. \\

$\lambda$ & Lagrange multiplier, interpreted as the subsidy for
keeping an arm passive. \\

$\mathcal B(\lambda)$ & Passive set at subsidy $\lambda$,
$\mathcal B(\lambda)=\{\omega\in\Pi:
\Delta(\omega,\lambda)\geq 0\}$. \\

$\delta(P^0)$ & Dobrushin contraction coefficient of the passive
transition matrix $P^0$. \\

$c$ & Contraction bound used in the analysis of the
belief-state value function. \\

$\Gamma_s$ & Posterior belief after activating the arm and observing state $s$\\

$D_{\Gamma}$ & denote the diameter of the post-activation belief\\

$D_{\Gamma,\tau}$ & Maximum discrepancy between the
activation-conditioned posterior beliefs and the passively
predicted belief, as defined in the indexability analysis. \\

$G_\lambda(\omega)$ & Subsidy sensitivity of the optimal value
function at belief $\omega$, where the derivative exists. \\

$L_R$ & Half the maximum reward span across the two actions\\

\hline
\end{tabular}
\end{table*}

\subsection{Belief Update}

For arm $n$, suppose  $a_n(t)=1$, observation $o_n(t)=k$ (state $s_n(t) = k$ is observed) is received, and the belief prior to the update is $\bm{\omega}_n(t),$ then the updated belief is $\bm{\omega}_n(t+1) = \Gamma_{n,k},$ for $k \in \{0,1,2\}$. Recall that $p_{n,ki}^{1}$ denotes the probability of transitioning from state $k$ to state $i$ under action $1$.
 So $\forall k \in \{0,1,2\},$  $\Gamma_{n,k} = [p_{n,k0}^{1}, p_{n,k1}^{1}, p_{n,k2}^{1}]^\top$. 
When $a_n(t) = 0$ for  arm $n$, the arm is not being interacted with, so no observation becomes available. In this case, we take the belief to evolve naturally on its own. Thus, $\bm{\omega}_n(t+1)  = (P_n^0)^{T} \bm{\omega}_n(t). $
\subsection{Problem Formulation}
A policy $\phi$ maps histories to actions:
$\phi: H(t)\to \{a_1,a_2,\dots,a_N\},$
where the history accumulated through time $t$ is $H(t)=\{\mathbf{a}(0),\mathbf{o}(0),\dots,\mathbf{a}(t-1),\mathbf{o}(t-1)\},$
using $\mathbf{a}(t)=\{a_1(t),\dots,a_N(t)\}$ actions of arms and $\mathbf{o}(t)=\{o_1(t),\dots,o_N(t)\}$ observation of arms.
$o_n(t)\in\mathcal O_n\cup\{\varnothing\},$ where,

$o_n(t)=
\begin{cases}
s_n(t), & a_n(t)=1,\\
\varnothing, & a_n(t)=0.
\end{cases}$

Under policy $\phi$, the total discounted reward over an infinite horizon with initial state $\mathbf{s} = (s_1,s_2, \cdots, s_N)$ is
\[
V_{\phi}(\mathbf{s}) = \mathbb{E}_{\phi}\!\left[\sum_{t=0}^{\infty} \sum_{n=1}^{N} \beta^{t}\, r_n(s_n(t),a_n(t))~\bigg \vert~ \mathbf{s}\right],
\]
where the budget constraint must hold at every time step
\[
\sum_{n=1}^{N} a_n(t) \leq B, \quad \forall\, t \geq 0,
\]
where $B$ represents  the available budget.
Arm $n$'s expected reward is given by
\[
R_n(\bm{\omega}_n(t),a_n(t)) 
= \sum_{s \in \mathcal{S}_n} \omega_n^s(t)\, r_n(s,a_n(t)).
\]
The set of admissible joint actions is
\begin{eqnarray*}
\mathcal{A} = \Big\{ \bm{a}(t)=(a_n(t))_{n=1:N} : a_n(t)\in\{0,1\},\; \\
\sum_{n=1}^N a_n(t)\leq B \Big\}.
\end{eqnarray*}

Given the belief state $\Omega=(\bm{\omega}_1,\dots,\bm{\omega}_N)$, the cumulative discounted value achieved under policy $\phi$ equals
\[
V_\phi(\Omega) = \mathbb{E}_\phi\!\left[\sum_{t=0}^{\infty}\sum_{n=1}^N \beta^t R_n(\bm{\omega}_n(t),a_n(t))~\bigg \vert~\Omega\right].
\]
The goal is to maximize the value function for given belief state and hence we have 
\[
V(\Omega)=\max_\phi V_\phi(\Omega).
\]
Recall that $\mathcal{O} \in \{0,1,2,\varnothing\}^N$.
The  dynamic programming formulation is as follows.
\begin{eqnarray*}
V(\Omega)=\max_{\bm{a}\in\mathcal{A}}
\Bigg[\sum_{n=1}^N R_n(\bm{\omega}_n,a_n)+\beta
\sum_{\bm{o}\in\mathcal{O}} V(\tau(\Omega,\bm{o},\bm{a}))\, \times \\
\Pr(\bm{o}\mid \Omega,\bm{a})\Bigg],
\end{eqnarray*}
with $\Omega'=\tau(\Omega,\bm{o},\bm{a})$.
Since the arms operate independently,
\[
\Pr(\bm{o}\mid\Omega,\bm{a})=\prod_{n=1}^N \Pr(o_n\mid \bm{\omega}_n,a_n),
\]
where $\Omega'=(\bm{\omega}_1',\dots,\bm{\omega}_N')$ and each component updates as $\bm{\omega}_n'=\tau(\bm{\omega}_n,o_n,a_n)$. Also, $\Pr(o_n=k\mid\bm\omega_n,1)=\omega_n^k$ and $\Pr(o_n=\varnothing\mid\bm\omega_n,0)=1.
$

This problem can be cast as a constrained POMDP (belief state MDP) problem by considering all the arms 
together with a resulting high dimensional state space and action space.  The computation of the optimal value function is hard \cite{papadimitriou1994computational}. 

\section{Decoupling of restless multi-armed bandit}
\label{sec:decoupled-RMAB}
We introduce a relaxed constraint version of the problem and then apply Lagrange relaxation. The relaxed constraint version of the problem is as follows.
\begin{align}\label{equation: Relaxed Problem}
    V(\Omega) = \max_{\phi}\mathbb{E}_\phi\!\left[\sum_{t=0}^{\infty}\sum_{n=1}^N \beta^t R_n(\bm{\omega}_n(t),a_n(t))\right], \nonumber \\
    \text{s.t.  } \mathbb{E}_{\phi}\left[\sum_{t=0}^{\infty}\beta^{t} \left(\sum_{n=1}^{N} a_n(t) \right)\right]\leq \frac{B}{1-\beta}.
\end{align}

Now, the Lagrangian relaxation of \eqref{equation: Relaxed Problem} is as follows.
\begin{align}\label{equation: L Relaxed Problem}
     V^{\lambda}(\Omega) = \max_{\phi}\mathbb{E}_\phi &\!\left[\sum_{t=0}^{\infty}\left(\sum_{n=1}^N \beta^t R_n(\bm{\omega}_n(t),a_n(t)) \right)\right] + \nonumber \\ &\lambda  \mathbb{E}_{\phi}\left[\sum_{t=0}^{\infty}\beta^{t} \left(B -  \sum_{n=1}^{N} a_n(t)  \right)\right],\nonumber \\
     & \text{s.t.  }\lambda \geq 0.  
\end{align}
From \eqref{equation: L Relaxed Problem}, the dynamic program formulation can be written as
{\small{
\begin{eqnarray}
    V^{\lambda}(\Omega) = \max_{\bm{a} \in \{0,1\}^N} \left\{ \sum_{n=1}^{N} R_n(\bm{\omega}_n,a_n) + \lambda\left(B - \sum_{n=1}^{N}a_n\right)
    \right. \nonumber \\ \left. 
    + \beta \sum_{\bm{o} \in \mathcal{O}} V^{\lambda}(\tau(\bm{\omega},\bm{o},\bm{a})) \prod_{n=1}^{N} \Pr(o_n ~|~\bm{\omega}_n,a_n)
    \right\}.
    \label{eqn:DP-1}
\end{eqnarray}
}}

Here, $\bm{a} = (a_1,\cdots,a_n,\cdots,a_N),$ $a_n \in \{0,1\}$ and $\lambda\geq 0.$ Note that $V^{\lambda}(\Omega)$ is the  optimal value function  for joint belief state $\Omega.$

Define the probability simplex  $\Pi = \{ \bm{\omega}_n: 0 \leq \bm{\omega}_n^s \leq 1,   \sum_{s} \bm{\omega}_n^s =  1 \}$  $1\leq  n \leq N.$

 From~\eqref{eqn:DP-1}, we obtain the following lemma.

\begin{lemma}[Decoupled RMAB]
The optimal value function $V^{\lambda}(\Omega)$ has the form:
\begin{eqnarray}
   V^{\lambda}(\Omega) = 
   \sum_{n=1}^{N} V_n^{\lambda}(\bm{\omega}_n) + \frac{(B-N) \lambda}{1-\beta},
   \label{eqn:DP-2}
\end{eqnarray}   
where, 
\begin{eqnarray}
    V_n^{\lambda}(\bm{\omega_n}) = 
    \max_{a_n \in \{0,1\}} 
    \left\lbrace
    R_n(\bm{\omega}_n,a_n) + \lambda(1- a_n) +
      \right. \nonumber \\ 
      \left. \beta \sum_{o_n \in \mathcal{O}_n}
    V_{n}^{\lambda}(\tau(\bm{\omega}_n,o_n,a_n)) \Pr(o_n~| \bm{\omega}_n,a_n)
    \right\},
    \label{eqn:val-func-lagrangge-SARB}
\end{eqnarray}
  $V_{n}^{\lambda}(\bm{\omega_n})$ is  the optimal value function for arm $n,$ $\bm{\omega}_n \in \Pi.$
\label{lemma:val-fun-decouple}
\end{lemma}
Note that this lemma decouples the $N$ armed bandit problem into $N$ single-armed restless bandits. The proof is in Appendix \ref{Proof-lemma:val-fun-decouple-appendix}.


\section{Single-armed restless bandit}
\label{sec:SAB}
We now study single-armed restless bandit and properties of value functions, and threshold type policy. For notational simplicity, we omit subscript $n$ in $\bm{\omega}_n$ and we represent belief state for a single armed using $\bm{\omega} \in \Pi.$ We also omit $n$ subscript for value functions. 
A belief state action value function for action $0$ and belief $\bm{\omega}$ is  
\begin{eqnarray}
    Q^{\lambda}(\bm{\omega},0 ) =  R(\bm{\omega},0) + \lambda + \beta V^{\lambda}(\tau(\bm{\omega})),
\end{eqnarray}
where $\tau(\bm{\omega}) = (P^0)^T \bm{\omega}.$
A belief state action value function for action $1$ and belief $\bm{\omega}$ is 
\begin{align}
    Q^{\lambda}(\bm{\omega},1 ) =  R(\bm{\omega},1) +  \beta  \sum_{s=0}^{2} \bm{\omega}^{s} V^{\lambda}(\Gamma_s) 
\end{align}
and the optimal value function is 
\begin{eqnarray}
    V^{\lambda}(\bm{\omega}) = \max \{ Q^{\lambda}(\bm{\omega},0 ),  Q^{\lambda}(\bm{\omega},1 ) \}.
\end{eqnarray}

For single-armed restless bandit problem, the Lagrangian parameter $\lambda$ is called as subsidy for not scheduling an arm.

\begin{lemma}[Convexity of value function]
	The  value function, $V^{\lambda}(\omega),$ is convex in $\omega$ for fixed $\lambda.$  
	\label{lemma:convexity}
\end{lemma}
The proof is through induction; it uses \cite[Lemma $2$]{Astrom69} to prove convexity of the value function. The proof of this is along lines of \cite{Meshram18}[Lemma~2].

\textbf{\textit{Monotone likelihood ratio}} (MLR) order is used for comparison of beliefs $\bm{\omega}$. For beliefs $
\bm{\omega},\bm{\omega}^{\prime}\in \Pi,
$ $\bm{\omega}$ dominates $\bm{\omega}'$ in the monotone likelihood-ratio order $\bm{\omega}\succeq_{\mathrm{MLR}}\bm{\omega}^{\prime},$ if $\omega^j\omega^{i,'} \geq \omega^i\omega^{j,'}$ for $ 0\leq i<j\leq 2.$

\textbf{\textit{Totally positive order of $\mathbf{2}$}} \textbf{($\mathrm{TP}2$)} is used for comparison of transition probability matrices for actions $0$ and $1.$ A nonnegative matrix $P=[p_{ij}]$ is said to be totally positive of order two,
or TP2, if every $2\times2$ minor is nonnegative: $p_{ij}p_{k\ell} -
p_{i\ell}p_{kj} \geq0 $  for all $
i<k,$ $  j<\ell.$ 
Suppose $P$ is TP2 and 
$
\bm{\omega}\succeq_{\mathrm{MLR}}\bm{\omega}'$
then
$
P^{\top}\bm{\omega}
\succeq_{\mathrm{MLR}}
P^{\top}\bm{\omega}'.
.$ 
If $P^{0}$ is TP2, then $\bm{\omega}\succeq_{\mathrm{MLR}}\bm{\omega}'$ implies $\tau(\bm{\omega})
\succeq_{\mathrm{MLR}}
\tau(\bm{\omega}').$
If $P^{1}$ is TP2, then a stronger ordering holds: $\Gamma_0
\preceq_{\mathrm{MLR}}
\Gamma_1
\preceq_{\mathrm{MLR}}
\Gamma_2.$

We next show that the optimal value function is monotone in belief. 
\begin{assum}
\begin{enumerate}
\item  
For every state-action pair,
\[
|r(i,a)|\leq r_{\max}<\infty,
\qquad
i\in\{0,1,2\},
\quad
a\in\{0,1\}.
\]
\item For each action $a\in\{0,1\}$, the state-dependent reward is non-decreasing. That is, $r(0,a)\leq r(1,a)\leq r(2,a).$

\item Define $d_i:=r(i,0)-r(i,1).$ 
Assume
$d_0\leq d_1\leq d_2.$

\item  $P^{0}$ and  $P^{1}$  are TP2.

\end{enumerate}
\label{assumption1}
\end{assum}

%
\begin{lemma}[Monotonicity of value function] 
	Under Assumption~\ref{assumption1}, the optimal value function $V^{\lambda}(\bm{\omega})$ is monotone in belief $\bm{\omega}$ that is, $V^{\lambda}(\bm\omega) \geq V^{\lambda}(\bm{\omega}')$ whenever $\bm{\omega}\succeq_{\mathrm{MLR}}\bm{\omega}^{\prime},$ for $\bm\omega, \bm{\omega}^{\prime} \in \Omega.$ 
	\label{lemma:monotonicity} 
\end{lemma} 
Proof sketch:  
Proof is using mathematical induction. First we have to show that $V_1^{\lambda}(\bm{\omega}) \geq V_1^{\lambda}(\bm{\omega}')$ whenever $\bm{\omega} \geq_{MLR}\bm{\omega}'.$ This follows from immediate reward assumption.  We assume that $V_t^{\lambda}(\bm{\omega}) \geq V_t^{\lambda}(\bm{\omega}')$ for some $t.$ Then we can show that $V_{t+1}^{\lambda}(\bm{\omega}) \geq V_{t+1}^{\lambda}(\bm{\omega}').$ For this we require TP2, MLR ordering and monotone difference in reward. Next one can show that $V_t^{\lambda} $ converges to $V^{\lambda}$ uniformly, and hence monotinicity property is preserved.
For detailed proof, see \cite[Proposition $1$]{Lovejoy87}. 

%

\begin{remark}
    Since $P^1$ is TP2 and the value function is monotone in belief. Thus we have $V^{\lambda}(\Gamma_0)
\leq V^{\lambda}(\Gamma_1)
\leq V^{\lambda}(\Gamma_2).$ $\mathcal{C}_{1}(\bm{\omega})
:=
\sum_{i=0}^{2}
\bm{\omega}^iV^{\lambda}(\Gamma_i).$ Then $\bm{\omega}
\succeq_{\mathrm{MLR}}
\bm{\omega}'$ implies $\mathcal{C}_{1}(\bm{\omega})
\geq
\mathcal{C}_{1}(\bm{\omega}').$ 
Define $\mathcal{C}_{0}(\bm{\omega})
:= V^{\lambda}\bigl(\tau(\bm{\omega})\bigr).$
Order preservation of $\tau$ gives $\tau(\bm{\omega})
\succeq_{\mathrm{MLR}} \tau(\bm{\omega}').$
Since $V^{\lambda}$ is increasing, $V^{\lambda}\bigl(\tau(\bm{\omega})\bigr)
\geq
V^{\lambda}\bigl(\tau(\bm{\omega}')\bigr).$
Hence, $\mathcal{C}_{0}(\bm{\omega})
\geq \mathcal{C}_{0}(\bm{\omega}').$
\end{remark}

We plan to show that the optimal policy is  threshold type.  

Define $\Delta(\bm{\omega},\lambda) := Q^{\lambda}(\bm{\omega},0) - Q^{\lambda}(\bm{\omega},1 ).$ We want to show that whenever 
$\bm{\omega}\succeq_{\mathrm{MLR}}\bm{\omega}^{\prime},$
we have $\Delta(\bm{\omega}\,\lambda) \geq \Delta(\bm{\omega}^{\prime},\lambda).$
A threshold type policy provides a partition of the belief state space $\Pi$ into three disjoint regions, $\Lambda_1, \Lambda_2, \Lambda_3 \subseteq \Pi,$  where $\Lambda_1 = \{\bm\omega \in \Pi: a^*_t(\bm\omega) = 1  \},$ $\Lambda_2 = \{\bm\omega \in \Pi: a^*_t(\bm\omega) =  0 \}$ and  $\Lambda_3 = \{\bm\omega \in \Omega(S): a^*_t(\bm\omega) = 1~\text{and }~ 0 \}.$ $a^*_t(\bm\omega) \in \{0,1\}$ is the optimal action for belief $\bm\omega$ at time step $t.$

\begin{definition}[Threshold type policy]	
	The optimal policy is called a threshold type if one of the following holds true. 
	\begin{enumerate}
		\item The optimal action $a_t^*(\bm\omega) =1$ for all $t$ and all $\bm\omega \in \Pi,$ that is $\Lambda_1 = \Pi$ and $\Lambda_2 = \Lambda_3 = \emptyset.$
		\item The optimal action $a_t^*(\bm\omega) =0$ for all $t$ and all $\bm\omega \in \Pi,$ that is $\Lambda_2 = \Pi$ and $\Lambda_1 = \Lambda_3 = \emptyset.$
		\item The optimal action $a_t^*(\bm\omega) =1$ for  all $\bm\omega \in \Lambda_1,$ $a_t^*(\bm\omega) =0$ for  all $\bm\omega \in \Lambda_2,$ and $a_t^*(\bm\omega) =1$ and $0$ for  all $\bm\omega \in \Lambda_3,$ that is, $\Lambda_1, \Lambda_2 ,\Lambda_3 \neq  \emptyset.$ Also, $\Lambda_i \cap \Lambda_j =\emptyset$, $i \ne j, i,j \in \{0,1,2\}$ and $\Lambda_1 \cap \Lambda_2 \cap \Lambda_3 = \emptyset.$ 
	\end{enumerate}
\end{definition}

Under Assumption~\ref{assumption1}, we want to  show that the following difference is either positive or negative. 
\begin{eqnarray}
\label{eq:single-crossing-condition}
&\Delta(\bm{\omega},\lambda) -  \Delta(\bm{\omega}^{\prime},\lambda) 
 = d(\bm{\omega})-d(\bm{\omega}') +  \\
&\beta
\left[
V^{\lambda}(\tau(\bm{\omega})) 
-
V^{\lambda}(\tau(\bm{\omega}^{\prime})) \right] 
- 
\beta
\sum_{i=0}^{2}
(\omega^{i}-\omega^{i,\prime})
V^{\lambda}(\Gamma_i). \nonumber
\end{eqnarray}
Then, there is a threshold type policy. 

\begin{remark}[Challenges in Threshold Policy]
 In general, even after Assumption~\ref{assumption1}, it is difficult to establish the threshold type policy result. It is due to negative sign of term $3$ in Eqn.~\eqref{eq:single-crossing-condition}. It is addressed by introducing parameter dependent Lipschitz constant and  we subsequently use it to provide sufficient conditions for Threshold-type policy.
\end{remark}



Let \(D_{\Gamma}\) denote the diameter of the post-activation belief set, $D_{\Gamma}
:=\max_{i,j\in\{0,1,2\}} \left\| \Gamma_i
-\Gamma_j \right\|_1.$  Let 
 $\operatorname{span}(r^a)
:= \max_{i\in\{0,1,2\}}r(i,a) -\min_{i\in\{0,1,2\}}r(i,a).$ This measures the maximum difference between the reward for different actions.  $L_R := \frac{1}{2}
\max_{a\in\{0,1\}} \operatorname{span}(r^a).$
\begin{proposition} $c
=
\max
\left\{
\delta(P^0),
\frac12
\max_{i,j}
\|
\Gamma_i-\Gamma_j
\|_1
\right\}$, where $\delta(P^0)$ is Dobrushin coefficient 
and $0 \leq  c \leq 1.$
\end{proposition}
\begin{proof}
    See Appendix \ref{proof of Proposition 1}.
\end{proof}
\begin{lemma}[Lipschitz bound on value function]
 \begin{equation}
\label{eq:three-state-lipschitz-main}
\left|
V^\lambda(\bm{\omega})
-
V^\lambda(\bm{\omega}')
\right|
\le
\frac{L_R}{1-\beta c}
\|
\bm{\omega}
-
\bm{\omega}'
\|_1, 
\end{equation}

\end{lemma}
The proof is given in the Appendix \ref{Appendix: Lemma 4}. 


\begin{remark}
 The Lipschtz bound, $L_{V} := \frac{L_R}{1-\beta c}$ depends on  span of rewards and transition probabilities through $D_{\Gamma}$ and $\delta(P^0).$ Under the appropriate model assumptions on transition probabilities, $c<1$ and further we can allow $\beta \rightarrow 1.$
\end{remark}

\begin{lemma}[MLR monotonicity of the reward difference]
\label{lem:d-mlr}

$\bm{\omega}=(\omega^0,\omega^1,\omega^2)^\top,$
define 
\begin{eqnarray*}
d(\bm{\omega})
:=
\sum_{i=0}^{2}d_i\omega^i
=
R(\bm{\omega},0)-R(\bm{\omega},1).
\end{eqnarray*}
Then $d(\bm{\omega})$ is nondecreasing with respect to $\bm{\omega}$ in MLR order. In particular 
$\bm{\omega}\succeq_{\mathrm{MLR}}\bm{\omega}'
\quad\Longrightarrow\quad
d(\bm{\omega})\geq d(\bm{\omega}').$
\end{lemma}
For the Proof, see the Appendix  \ref{Appendix:Lemma 5}.

We next provide sufficiency conditions for threshold type policy. We require condition on $d(\bm{\omega}) - d(\bm{\omega}').$ This is stated in next theorem.




\begin{theorem}[Sufficient condition for a threshold policy]
\label{thm:threshold-condition}
Suppose Assumption $1$ holds.
Let
\[
m_d:=
\inf_{\substack{\bm{\omega}\geq_r\bm{\omega}'\\
\bm{\omega}\neq\bm{\omega}'}}
\frac{
2\left[d(\bm{\omega})-d(\bm{\omega}')\right]
}{
\|\bm{\omega}-\bm{\omega}'\|_1
},
\]
and  $c=
\max\left\{
\delta(P^0),\frac{D_\Gamma}{2}
\right\}.$
If $m_d
\geq
\frac{\beta L_R D_\Gamma}
{1-\beta c},$
then, for every
$\bm{\omega}\succeq_{MLR}\bm{\omega}'$,
$\Delta(\bm{\omega},\lambda)
\geq
\Delta(\bm{\omega}',\lambda).$
Consequently, the action advantage is MLR nondecreasing,
which yields a threshold-type optimal policy.
\end{theorem}

The Proof is in the Appendix \ref{Appendix: Theorem 1}.


\begin{remark}[As $\beta\to1^{-}$]
If $c<1$, then
$\frac{\beta L_RD_\Gamma}{1-\beta c }
\longrightarrow
\frac{L_RD_\Gamma}{1-c}<\infty,
\quad\text{as }\beta\to1^-.$
Hence, the sufficient condition remains bounded in the nearly
undiscounted regime. In contrast, if $c=1$ and
$L_RD_\Gamma>0$, the bound diverges as $\beta\to1^-$.
Thus, strict contraction $c<1$ ensures a nondegenerate
sufficient condition as the discount factor approaches one.
\end{remark}

\begin{remark}[Strict contraction $c< 1$]
\begin{enumerate}
\item If $\delta(P^0)<1$ and $D_\Gamma<2,$
then $c<1, $ where $c=
\max\left\{
\delta(P^0),\frac{D_\Gamma}{2}
\right\}.$
\item Sufficient condition for $D_\Gamma<2:$  
If every pair
$\Gamma_i,\Gamma_j$ has overlapping support, i.e.,
$\sum_k\min\{\Gamma_i^k,\Gamma_j^k\}>0,
\quad \forall\, i,j,$ then $D_\Gamma<2.$
\item Sufficient condition for strict Dobrushin contraction: If $P^0_{ij}>0,$  $\forall i,j\in\{0,1,2\},$ i.e., $P^0$ is strictly positive, then $\delta(P^0)<1.$

\end{enumerate}
\end{remark}

\begin{corollary}
After activation the state is perfectly observed, $\bm{\Gamma}_i=(p_{i,\cdot}^{1})^{T}=\begin{bmatrix}
p_{i0}^{1}\\
p_{i1}^{1}\\
p_{i2}^{1}
\end{bmatrix}.$
$ D_\Gamma = \max_{i,j\in\{0,1,2\}} \sum_{k=0}^{2}
\left| p_{ik}^{1}-p_{jk}^{1} \right|.$ Then,
second sufficiency condition  is 
\begin{equation}
\label{eq:fully-explicit-increasing-condition}
\begin{aligned}
&
m_d
\quad\geq
\frac{
\displaystyle
\beta
\max_{i,j}
\sum_{k=0}^{2}
\left|
p_{ik}^{1}-p_{jk}^{1}
\right|
}{
\displaystyle
2\left(
1-
\frac{\beta}{2}
\max_{\substack{a\in\{0,1\}\\i,j}}
\sum_{k=0}^{2}
\left|
p_{ik}^{a}-p_{jk}^{a}
\right|
\right)
}
\\[1mm]
&\qquad\qquad\times
\max_{a\in\{0,1\}}
\left[
\max_i r(i,a)-\min_i r(i,a)
\right].
\end{aligned}
\end{equation}
\end{corollary}

\subsection{Birth-death model}
For $M=3$, the birth-death transition matrix under action
$a\in\{0,1\}$ is
\[
P^a=
\begin{bmatrix}
1-p_a & p_a & 0\\
q_a & 1-p_a-q_a & p_a\\
0 & q_a & 1-q_a
\end{bmatrix},
\]
where $p_a,q_a\geq0,
\quad
p_a+q_a\leq1.$

\begin{theorem}[Three-state birth-death model]
\label{lem:three-state-bd}
Consider a three-state birth-death model with TP 2 transition
probabilities $p_a$ and $q_a$ under action $a\in\{0,1\}$.
Let $m_P:=\min\{p_0,q_0,p_1,q_1\}>0.$ If $m_d
\geq
\frac{2 \beta\bigl(1-\min\{p_1,q_1\}\bigr)}
     {1-\beta(1-m_P)}\, L_R,$
then under Assumption 1, the optimal policy has a threshold-type
structure with respect to the MLR order.
\end{theorem}
The proof is in the Appendix \ref{Proof of Theorem 2}.

\subsection{Restart Model}

\begin{theorem}[Threshold policy for a three-state restart model]
\label{lem:restart-threshold}
Suppose at action $0$ belief evolves according to $P^0$, while action $1$
restarts the process according to a fixed distribution $\gamma$, i.e., $P^1_{i,\cdot}=\gamma,\quad i \in\{0,1,2\}.$ Then, under Assumption~1, the optimal policy is threshold-type with respect to the MLR order.
\end{theorem}

\begin{proof}
    From theorem \ref{thm:threshold-condition}, $m_d
\geq
\frac{\beta L_R D_\Gamma}
{1-\beta c}$ implies a threshold type policy. Clearly, in this case, $D_\Gamma=0$ since $P^1_{i,\cdot}=\gamma,\quad i\in\{0,1,2\}.$ Therefore, the sufficient condition reduces to $
m_d\geq 0$, which holds under Assumption~1. 
\end{proof}

\begin{remark}[Comparison with work of  \cite{Lovejoy87}]
In \cite{Lovejoy87}, author established MLR-monotonicity of optimal policies under
general structural conditions on rewards and transition kernels.
For the three-state restart model, our result exploits the
identical-row structure of the restart action and gives the explicit
sufficient condition $m_d\geq \frac{2\beta c_0}{1-\beta c_0}L_R$, where $c_0=\delta(P^0).$ Thus, our condition provides a directly verifiable quantitative
criterion for an MLR-monotone threshold policy in the three-state
restart setting.
\end{remark}

\section{Whittle Index Policy and Index Computation Scheme}
\label{sec:Whittle-index}

We now study Whittle index policy, \cite{Whittle88}, where arms are scheduled in each time slot based on heuristic index assigned to an arm, that maps belief state to an real number called index.  This policy is shown to  be near optimal \cite{LiuZhao10} for two-state model.  The key criteria for index policy is the indexability, and one require to show that each arm is indexable. The notion  of indexability for each arm is defined using the passive set $\mathcal{B}(\lambda).$ It is set of all belief states for which not scheduling an arm is optimal action, i.e., it is the set of belief states for which passive action is optimal. Using this we define indexability.
\begin{eqnarray}
    \mathcal{B}(\lambda) :=  
    \{ \bm{\omega} \in \Omega:    Q_n^{\lambda}(\bm{\omega}_n,0 ) \geq     Q_n^{\lambda}(\bm{\omega}_n,1 )   \}
\end{eqnarray}

\begin{definition}[Indexability]
    A single-armed restless bandit is indexable if $ \mathcal{B}(\lambda) $ is monotonically increases from $\emptyset$ to $\Pi$ (full belief simplex) as $\lambda$ increases from $-\infty$ to $+\infty.$ i.e., if $\lambda_1 \leq \lambda_2$ then $\mathcal{B}(\lambda_1) \subseteq \mathcal{B}(\lambda_2).$
    An RMAB is indexable, if each arm is indexable.
\end{definition}

\begin{remark}
Indexability requires that the passive set $\mathcal{B}(\lambda)$
expands monotonically with the subsidy $\lambda$. Thus, once it is
optimal to choose the passive action at a belief $\bm{\omega}$ for
some subsidy $\lambda$, it remains optimal to be passive at
$\bm{\omega}$ for every larger subsidy. Thus, increasing
$\lambda$ cannot cause a belief to switch from passive back to
active.
\end{remark}
Using indexable definition, we define Whittle index, 

\begin{definition}(Whittle's index)
If an arm is indexable with  belief state $\bm{\omega} \in \Pi,$ then Whittle's index $W(\bm{\omega})$ is the minimum subsidy for which, not scheduling the arm will be the optimal action., i.e., 
 \begin{equation*}
  W(\bm{\omega}) = \inf\{\lambda \in\mathbb{R}:   Q^{\lambda}(\bm{\omega},0 ) \geq     Q^{\lambda}(\bm{\omega},1 ) \}
 \end{equation*}
\end{definition}

\begin{remark}
\begin{enumerate}
    \item For an indexable arm, $W(\bm{\omega})$ represents the critical
subsidy for passivity at belief $\bm{\omega}$. At
$\lambda=W(\bm{\omega})$, the two actions are typically indifferent,
$Q^{W(\bm{\omega})}(\bm{\omega},0)
=
Q^{W(\bm{\omega})}(\bm{\omega},1).$
If $\lambda<W(\bm{\omega}),$ then $Q^\lambda(\bm{\omega},1)>Q^\lambda(\bm{\omega},0),$ and scheduling the arm is optimal. If $\lambda>W(\bm{\omega})$ then $\lambda>W(\bm{\omega})
\quad\Longrightarrow\quad
Q^\lambda(\bm{\omega},0)\geq Q^\lambda(\bm{\omega},1),$  keeping the arm passive is optimal.

\item The main challenges are to establish indexability and to compute the
Whittle index. Indexability requires proving
\[
\mathcal{B}(\lambda_1)\subseteq\mathcal{B}(\lambda_2),
\qquad \lambda_1\leq\lambda_2,
\]
over the entire belief space. Moreover,
\[
W(\bm{\omega})
=\inf\{\lambda:
Q^\lambda(\bm{\omega},0)\geq Q^\lambda(\bm{\omega},1)\}
\]
is generally implicit, since $Q^\lambda$ itself depends on $\lambda$
through the optimal value function. Thus, even for an indexable arm,
a closed-form Whittle index is generally difficult to obtain.
\item To establish indexability, it is sufficient to show that
$\Delta(\bm{\omega},\lambda)$ is nondecreasing in $\lambda$. For
$\lambda_2>\lambda_1$, define
\[
D_{\lambda_1,\lambda_2}(\bm{\omega})
:=
V^{\lambda_2}(\bm{\omega})
-
V^{\lambda_1}(\bm{\omega}).
\]
Then
\begin{align}
\Delta(\bm{\omega},\lambda_2)
-\Delta(\bm{\omega},\lambda_1)
&=
(\lambda_2-\lambda_1)
+\beta D_{\lambda_1,\lambda_2}(\tau(\bm{\omega})) \nonumber\\
&\quad
-\beta\sum_{i=0}^{2}\omega_i
D_{\lambda_1,\lambda_2}(\Gamma_i).
\label{eq:Delta-lambda-difference}
\end{align}
Although $\lambda_2-\lambda_1>0$, the continuation terms need not have
a definite sign. Hence, monotonicity of
$\Delta(\bm{\omega},\lambda)$ in $\lambda$ does not follow directly
and requires additional structural conditions.
\end{enumerate}
\end{remark}


Suppose that $\Delta(\bm{\omega},\lambda)$ is differentiable almost everywhere. At every differentiability point,
\[
\frac{\partial}{\partial\lambda}
\Delta(\bm{\omega},\lambda)
=
1+\beta G_\lambda(\tau(\bm{\omega}))
-\beta\sum_{i=0}^{2}\omega_iG_\lambda(\Gamma_i),
\]
where $G_\lambda(\bm{\omega})
:=
\frac{\partial V^\lambda(\bm{\omega})}{\partial\lambda}.$
Since $V^\lambda(\bm{\omega})$ is nondecreasing and convex in
$\lambda$, it is differentiable almost everywhere and  subgradients exists at every
$\lambda$. Under the envelope interpretation \cite{milgrom2002envelope}, whenever the
derivative is uniquely defined,
\[
G_\lambda(\bm{\omega})
=
\mathbb{E}_{\bm{\omega}}^{\phi_\lambda^\star}
\left[
\sum_{t=0}^{\infty}\beta^t
\mathbf{1}\{A_t=0\}
\right],
\]
and hence $0\leq G_\lambda(\bm{\omega})\leq\frac{1}{1-\beta}.$ 
\begin{lemma}[Uniform Lipschitz Continuity of Subsidy Sensitivity]
We have 
\[
\left|
G_\lambda(\bm{\omega})
-G_\lambda(\bm{\omega}')
\right|
\leq
L_G\|\bm{\omega}-\bm{\omega}'\|,
\quad
\forall\,\bm{\omega},\bm{\omega}'\in\Pi,
\]
where $L_G= \frac{\beta}{2(1-\beta)}.$
\label{Lemma:Uniform-Lip-sensitivity}
\end{lemma}
The proof of Lemma is in the Appendix \ref{Proof:Lemma-Uniform-Lip-sensitivity}.

At a nondifferentiability point, $G_\lambda$ may be replaced by an
appropriate subgradient, corresponding to
an optimal policy at that subsidy. We can establish same Lipschitz bound and same $L_G.$  We skip the details and it is provided in the Appendix.  Define $D_{\Gamma,\tau}
:=
\sup_{\bm{\omega}\in\Pi}
\sum_{s=0}^{2}\omega^s
\left\|
\Gamma_s-\tau(\bm{\omega})
\right\|.$

\begin{theorem}[Sufficient condition for indexability]
\label{lem:lipschitz-indexability}
Suppose that 
\[
\left|
G_\lambda(\bm{\omega})
-G_\lambda(\bm{\omega}')
\right|
\leq
L_G\|\bm{\omega}-\bm{\omega}'\|,
\quad
\forall\,\bm{\omega},\bm{\omega}'\in\Pi
\]
with $L_G = \frac{\beta}{2(1-\beta)}.$
If $\beta L_GD_{\Gamma,\tau}\leq1,$ 
then $\Delta(\bm{\omega},\lambda)$ is nondecreasing in $\lambda$
for every $\bm{\omega}\in\Pi$. Consequently, the passive sets
$\mathcal{B}(\lambda)$ are nondecreasing in $\lambda$, and the arm
is indexable.
\end{theorem}
\textbf{Proof:}
By the Lipschitz assumption,
\[
G_\lambda(\Gamma_s)
-G_\lambda(\tau(\bm{\omega}))
\leq
L_G\|\Gamma_s-\tau(\bm{\omega})\|.
\]
Therefore,
\begin{align*}
\frac{\partial}{\partial \lambda}\Delta(\bm{\omega},\lambda)
&\geq
1-\beta L_G
\sum_{s=0}^{2}\omega^s
\|\Gamma_s-\tau(\bm{\omega})\|\\
&\geq
1-\beta L_GD_{\Gamma,\tau}
\geq0.
\end{align*}
Hence, for each fixed $\bm{\omega}$,
$\Delta(\bm{\omega},\lambda)$ is nondecreasing in $\lambda$.


Since
$\mathcal{B}(\lambda)
=
\{\bm{\omega}\in\Pi:
\Delta(\bm{\omega},\lambda)\geq0\},$
we obtain
$\lambda_1\leq\lambda_2
\quad\Longrightarrow\quad
\mathcal{B}(\lambda_1)
\subseteq
\mathcal{B}(\lambda_2).$
Together with the usual limiting conditions $\mathcal{B}(\lambda)\to\emptyset
\quad\text{as }\lambda\to-\infty,$
$\mathcal{B}(\lambda)\to\Pi
\quad\text{as }\lambda\to+\infty.$
the arm is indexable.


Using the bound $L_G
\leq
\frac{\beta}
{2(1-\beta)},$ the sufficient condition $\beta L_GD_{\Gamma,\tau}\leq1$ yields
\[
\beta^2D_{\Gamma,\tau}
\leq
2(1-\beta)
\]
Note that $0\leq D_{\Gamma,\tau} \leq 2.$ The proof is provided in Appendix \ref{Appendix: Bound}.

\begin{proposition}[Admissible discount factor]
\label{prop:beta-admissible-range}
Let $D_{\Gamma,\tau}\in[0,2]$.
If $D_{\Gamma,\tau}>0$, the sufficient condition $\beta^2D_{\Gamma,\tau}\leq 2(1-\beta)$
holds if and only if
\begin{equation}
\label{eq:beta-maximum}
0<\beta\leq
\beta_{\max}(D_{\Gamma,\tau})
:=
\frac{2}{1+\sqrt{1+2D_{\Gamma,\tau}}}.
\end{equation}
If $D_{\Gamma,\tau}=0$, the condition holds for every
$\beta\in(0,1)$.
\label{Prop:admissible-discount-parameter}
\end{proposition}

Proof is in the Appendix \ref{Proof of Proposition 2}.

Note that the estimate of $L_G$ uses
the global bound
$0\leq G_\lambda(\bm{\omega})\leq\frac{1}{1-\beta}.$
Since indexability depends only on differences of $G_\lambda$, a
sharper span or contraction estimate may be used.


\begin{remark}[Behavior of the sufficient condition as $\beta\uparrow1$]
\label{rem:indexability-beta}
\begin{enumerate}
\item  
The sufficient condition 
can equivalently be written as
$D_\Gamma
\leq
\frac{2(1-\beta)}{\beta^2}.$
Hence,  $D_{\Gamma,\tau}
=
O(1-\beta),$ $\beta\uparrow1.$

\item  If $D_{\Gamma,\tau}>0$ is fixed independently of $\beta$,
the above sufficient condition cannot hold for all discount factors
arbitrarily close to one. In the limiting case $D_{\Gamma,\tau}=0$, however,
the condition is satisfied for every $\beta\in(0,1)$.

\item 
The quantity $D_{\Gamma,\tau} =
\sup_{\bm{\omega}\in\Pi}
\sum_{s=0}^{2}\omega^s
\|\Gamma_s-\tau(\bm{\omega})\|$
measures the maximal discrepancy between the observation-conditioned
beliefs $\Gamma_s$ and the passively predicted belief
$\tau(\bm{\omega})$. Thus, a small $D_{\Gamma,\tau}$ corresponds to a model
in which the observation/update mechanism does not move the belief
too far from its passive prediction. Consequently, transition and
observation parameters yielding closely spaced posterior beliefs
make the sufficient indexability condition easier to satisfy.
\end{enumerate}
\end{remark}

\subsection{Whittle Index Computation Algorithm}

Note that closed form expressions for value functions are not available.
Hence, it is difficult to obtain a closed form expression for the index. For a fixed belief state $\bm{\omega}$, we have  the action-value
difference $\Delta(\bm{\omega},\lambda)
=
Q^{\lambda}(\bm{\omega},0)
-
Q^{\lambda}(\bm{\omega},1).$
Recall that the Whittle index $W(\bm{\omega})
=
\inf\{\lambda:\Delta(\bm{\omega},\lambda)\ge0\}.$ Therefore, $\Delta(\bm{\omega},W(\bm{\omega}))=0.$

Thus, computation of the Whittle index can be viewed as a root-finding
problem in the subsidy parameter $\lambda.$ 

For a given subsidy $\lambda_t$, the value function is first computed
using value iteration, and the corresponding action values $Q^{\lambda_t}(\bm{\omega},0),$ $Q^{\lambda_t}(\bm{\omega},1)$ 
are evaluated. and hence $\Delta(\bm{\omega},\lambda_t)
=
Q^{\lambda_t}(\bm{\omega},0)
-
Q^{\lambda_t}(\bm{\omega},1)$
If $|\Delta(\bm{\omega},\lambda_t)|\leq h$, where $h>0$ is a prescribed tolerance, then
$\lambda_t$ is taken as an approximation of $W(\bm{\omega})$.
Otherwise, the subsidy is updated according to
\[
\lambda_{t+1}
=
\lambda_t-\alpha_t \Delta(\bm{\omega},\lambda_t),
\]
where $\{\alpha_t\}$ is a sequence of step sizes such that $\alpha_t>0,$ $\sum_{t=0}^{\infty}\alpha_t=\infty$ and $\sum_{t=0}^{\infty}\alpha_t^2<\infty.$
The value iteration procedure is as follows.
\begin{eqnarray*}
     Q_{t+1}^{\lambda_t}(\bm{\omega},0 ) &=&  R(\bm{\omega},0) + \lambda_t + \beta V_t^{\lambda_t}(\tau(\bm{\omega})), 
    \\
     Q_{t+1}^{\lambda_t}(\bm{\omega},1 ) &=&  R(\bm{\omega},1) +  \beta  \sum_{s=0}^{2} \bm{\omega}^{s} V_{t}^{\lambda_t}(\Gamma_s) \\
    V_{t+1}^{\lambda_t}(\bm{\omega}) &=& \max\{ Q_{t+1}^{\lambda_t}(\bm{\omega},0 ), Q_{t+1}^{\lambda_t}(\bm{\omega},1 )\}    
\end{eqnarray*}


The value-iteration procedure used at each subsidy value also requires
an inner convergence tolerance $\epsilon_{\mathrm{VI}}$.  The value iteration may be terminated when $\left\|
V_{t+1}^{\lambda_t }-V_t^{\lambda_t}
\right\|_{\infty}
\leq \epsilon_{\mathrm{VI}}.$ Here, $\epsilon_{\mathrm{VI}}$ is chosen sufficiently
smaller than the  index tolerance $h$ so that numerical errors in
the computed $Q$-values do not dominate the root-finding criterion. 

The mathematical analysis of $\lambda_t$ update is done using two-timescale algorithm, where 
on slower timescale we update the subsidy $\lambda_t$ and on faster timescale we run value iteration $Q_{t}^{\lambda_t}(\bm{\omega},0),$ $Q_{t}^{\lambda_t}(\bm{\omega},1).$  The convergence analysis follows from stochastic approximation algorithms, see \cite[Section 9.3] {Borkar08} and \cite[Section III, B]{Borkar2018}.  
The Whittle index for belief state $\bm{\omega},$ is $\lambda = \lambda(\bm{\omega})$  such that $\vert Q_{t}^{\lambda_t = \lambda}(\bm{\omega},0) - Q_{t}^{\lambda_t = \lambda}(\bm{\omega},1)  \vert \leq h$ and $\left\|
V_{t+1}^{\lambda_t = \lambda}-V_t^{\lambda_t=\lambda}
\right\|_{\infty}
\leq \epsilon_{\mathrm{VI}}.$

We notice that the preceding algorithm of index computation is computationally expensive. It depends on the stepsize sequence $\alpha_t$ and it is slow. We next discuss another alternative for index computation bisection method which is computationally efficient compared to stochastic approximation based algorithm. 

\begin{algorithm}[h]
\caption{Whittle Index Computation Using Bisection}
\label{algo:WI-bisection}
\KwIn{Belief state $\bm{\omega}$, subsidy bounds
$\lambda_L,\lambda_U$, tolerance $h$, and value-iteration
tolerance $\epsilon_{\mathrm{VI}}$.}
\KwOut{Whittle index $W(\bm{\omega})$.}


Compute $\Delta(\bm{\omega},\lambda_L)$ and
$\Delta(\bm{\omega},\lambda_U)$ using value iteration\;

\If{$\Delta(\bm{\omega},\lambda_L)
      \Delta(\bm{\omega},\lambda_U)>0$}{
\textbf{stop}: the initial interval does not bracket an index\;
}

\While{$\lambda_U-\lambda_L>h$}{
Set
\[
\lambda_M\gets\frac{\lambda_L+\lambda_U}{2}.
\]

Compute
$Q^{\lambda_M}(\bm{\omega},0)$ and
$Q^{\lambda_M}(\bm{\omega},1)$ using value iteration with
tolerance $\epsilon_{\mathrm{VI}}$\;

Set
\[
\Delta(\bm{\omega},\lambda_M)\gets
Q^{\lambda_M}(\bm{\omega},0)
-
Q^{\lambda_M}(\bm{\omega},1).
\]

\If{$|\Delta(\bm{\omega},\lambda_M)|\leq h$}{
$W(\bm{\omega})\gets\lambda_M$\;
\Return{$W(\bm{\omega})$}\;
}

\If{$\Delta(\bm{\omega},\lambda_L)\Delta(\bm{\omega},\lambda_M)\leq0$}{
$\lambda_U\gets\lambda_M$\;
}
\Else{
$\lambda_L\gets\lambda_M$\;
}
}

\[
W(\bm{\omega})
\gets
\frac{\lambda_L+\lambda_U}{2}.
\]

\Return{$W(\bm{\omega})$}\;
\end{algorithm}

\paragraph{Bisection method for index computation.}
For a fixed belief $\bm{\omega}$, the Whittle index $W(\bm{\omega})$ is $\Delta(\bm{\omega}, W(\bm{\omega})) = 0.$ 
We start with an interval $[\lambda_L,\lambda_U]$ that satisfies $\Delta(\bm{\omega},\lambda_L)\leq 0
$ and $\Delta(\bm{\omega},\lambda_U)\geq0.$ 

At each iteration, the midpoint $\lambda_M=\frac{\lambda_L+\lambda_U}{2}$ is computed. For each $\lambda_M$, the corresponding value function
and $Q$-values are computed using value iteration. If $\Delta(\bm{\omega},\lambda_M)<0,$, choose $a=1$ and the lower endpoint is updated as 
\[
\lambda_L\leftarrow\lambda_M.
\]
If $\Delta(\bm{\omega},\lambda_M)\geq 0$, choose $a=0$ and the upper endpoint is updated as 
\[
\lambda_U\leftarrow\lambda_M.
\]
Thus, the search interval is reduced by a factor of two at each iteration. The procedure ends when $|\Delta(\bm{\omega},\lambda_M)|\leq h_Q$ or $ \lambda_U-\lambda_L\leq h_\lambda$, where $h_Q$ and $h_\lambda$ are prescribed tolerances. The resulting midpoint is then used as the numerical approximation of $W(\bm{\omega})$. This procedure is particularly suitable when
indexability ensures the required monotonicity of the action-value
difference with respect to the subsidy.


The detailed algorithm is provided in \ref{algo:WI-bisection}.

\begin{remark}[Computational complexity]
Let $C_{\mathrm{VI}}$ denote the cost of solving the single-arm
dynamic program for a fixed subsidy. Bisection computes the index
with subsidy accuracy $h_\lambda$ in
\[
O\!\left(
C_{\mathrm{VI}}
\log\frac{\lambda_U-\lambda_L}{h_\lambda}
\right)
\]
operations per belief state. In contrast, a stochastic-approximation
scheme requiring $N_{\mathrm{SA}}$ iterations has complexity
$O(N_{\mathrm{SA}}C_{\mathrm{VI}})$. Thus, under indexability and
accurately computed action values, bisection provides a predictable
logarithmic search complexity, whereas stochastic approximation is
more suitable when action values are obtained from noisy or
sample-based estimates.
\end{remark}











\section{Simulation Results}
\label{sec:sim-example}

In this section, we first present the threshold type policy and then compare the performance of the Whittle index policy, myopic policy, and round-robin policy for homogeneous arms and heterogeneous arms. All simulations are done for the three state  model.

\subsection{Demonstration of threshold-type policy} 
The passive and active transition probability matrices are
\[
P^{p}=
\begin{bmatrix}
0.5 & 0.4 & 0.1\\
0.2 & 0.5 & 0.3\\
0.1 & 0.2 & 0.7
\end{bmatrix},
\qquad
P^{a}=
\begin{bmatrix}
0.4 & 0.4 & 0.2\\
0.1 & 0.4 & 0.5\\
0.05 & 0.15 & 0.8
\end{bmatrix}.
\]
The reward vectors under active and passive actions are  $R^{a}=[0,\;1,\;2]$ and $R^{p}=[0,\;0.8,\;1.6]$, respectively.
We use a discount factor of $\beta=0.99$. We discretize the belief simplex using a grid resolution of $80$ and solve the discounted Bellman equation by value iteration until tolerance of $10^{-6}$ is reached  or a maximum of $2000$ iterations. 
In Figure~\ref{fig:policy_subsidy}, the optimal policy is displayed for  belief states.  We compute the optimal policy for subsidy values $\lambda = 0.4$ and $\lambda = 0.5$ and we visualize the resulting active and passive regions on the belief simplex.
These experiments illustrate how the switching boundary evolves with the subsidy and demonstrate the monotonic expansion of the passive region as the subsidy increases.

The blue points represent belief states belonging to the passive set $\mathcal{B}(\lambda)$ at subsidy $\lambda$, while the red points represent belief states in the active set. As the subsidy $\lambda$ increases, the passive set $\mathcal{B}(\lambda)$ expands, illustrating the indexability property. This behavior is consistent with the indexability property of the model.



\begin{figure}[t]
\begin{tabular}{cc}
    \centering
    \hspace{-0.02\textwidth}
    \begin{subfigure}{0.24\textwidth}
        \centering
        \includegraphics[width=\linewidth]{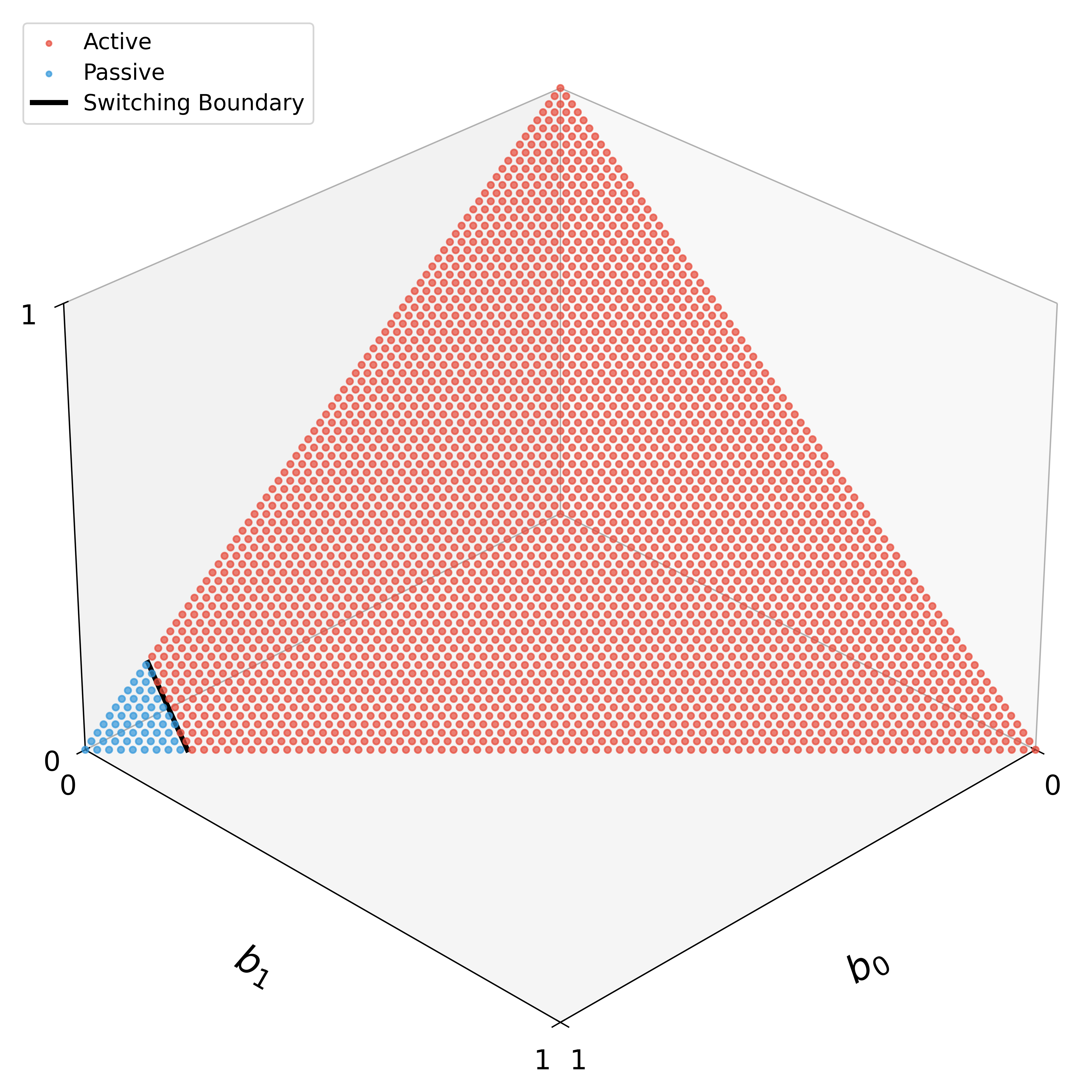}
        \caption{$\lambda=0.4$}
    \end{subfigure}
    &
    \hspace{-0.02\textwidth}
    \begin{subfigure}{0.24\textwidth}
        \centering
        \includegraphics[width=\linewidth]{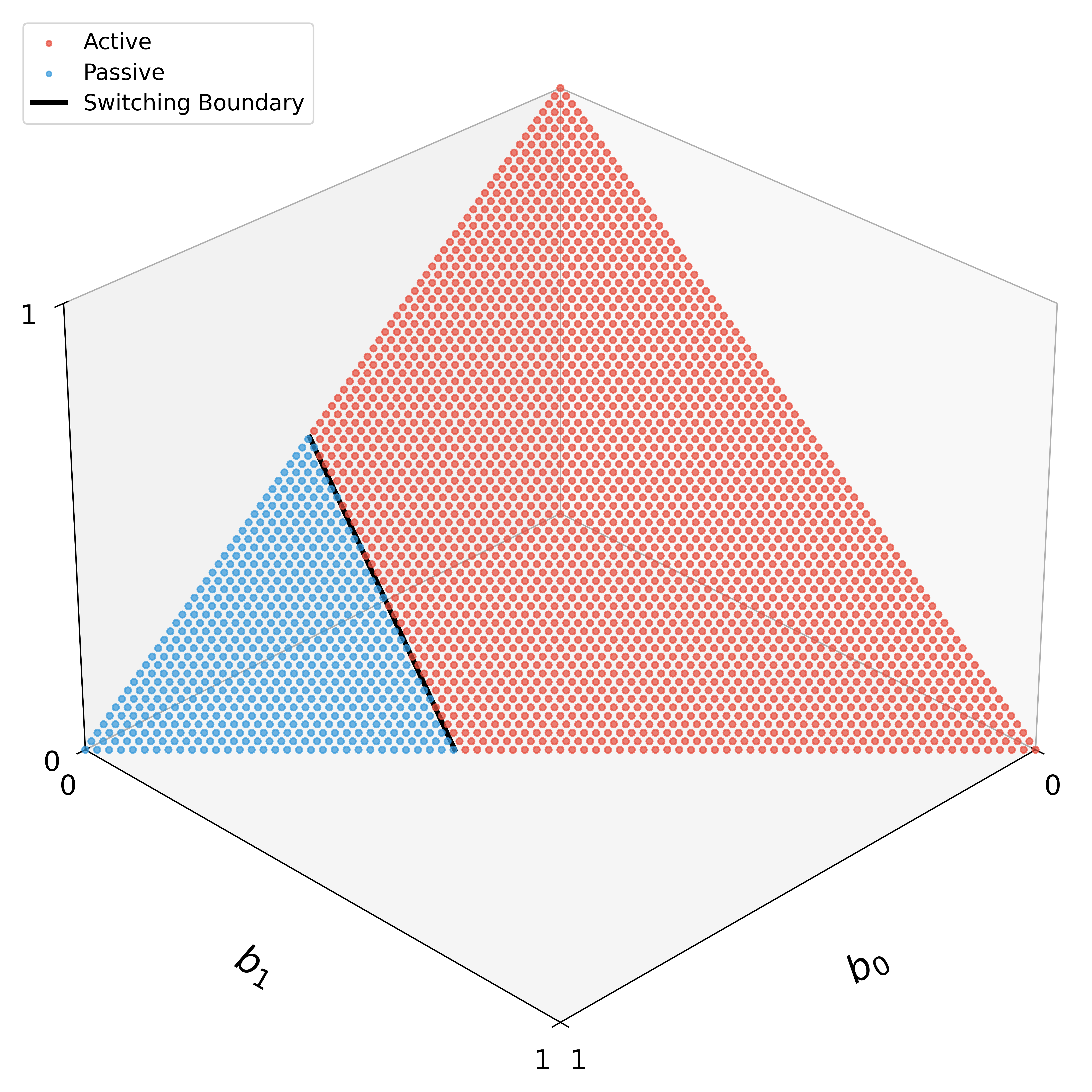}
        \caption{$\lambda=0.5$}
    \end{subfigure}
\end{tabular}
    \caption{Optimal policy regions for different subsidy values. The blue region denotes the passive action, while the red region denotes the active action. The black curve represents the switching boundary between the two actions. As the subsidy increases, the passive region expands, indicating that more belief states are optimally left passive.}
    \label{fig:policy_subsidy}
\end{figure}
\subsection{Performance evaluation of different policies}



We compare the discounted cumulative reward as function of iteration of the Whittle index policy, a myopic greedy policy, and a round-robin policy 
At each decision epoch, exactly $B$ arms are activated while the remaining arms evolve passively.

\textbf{The Whittle index policy:} Here, the index of each arm is computed using index computation algorithm.  Let $W_n(\bm{\omega}_n(t))$ denote the Whittle index of arm $n$ at belief $\bm{\omega}_n(t)$. The Whittle index policy activates the $B$ arms with the largest indices. 

\textbf{Myopic policy:} Here, index for each arm is defined using immediate expected reward activation advantage function function, it is defined as follows $G_n(t)=\omega_n(t)^\top(R_n^1-R_n^0)$ where $\omega_n(t)$ denote the belief state of arm $n$ at time $t.$ 
$B$ arms with the highest $G_n(t).$  This is also referred to as Greedy policy.

\textbf{Round robin policy:} Here, $B$  arms are activated  according to a fixed cyclic ordering, independently of their belief states.

\subsubsection{Homogeneous arms:}
In this, we consider identical and homogeneous.
In our experiment, we set
the passive and the active transition matrices as  
\[P^{p}=
\begin{bmatrix}
0.90 & 0.08 & 0.02\\
0.40 & 0.45 & 0.15\\
0.35 & 0.45 & 0.20
\end{bmatrix},\ \
P^{a}=
\begin{bmatrix}
0.05 & 0.15 & 0.80\\
0.02 & 0.13 & 0.85\\
0.01 & 0.04 & 0.95
\end{bmatrix}.
\]
 The immediate reward vectors under active and passive action are $R^{a}=[0,\;1,\;2]$ and $R^{p}=[0,\;1.1,\;1.9]$, respectively.

We solve the Whittle index using the subsidy-based value iteration algorithm with discount factor $\beta=0.99$, belief grid resolution $40$, convergence tolerance $10^{-6}.$
 We use number of arms $N=50$ statistically identical arms over a horizon of $T=500$ decision epochs, where exactly $B=10$ arms are activated at each time step. 
 

Figures \ref{subfig:Average_minmax_homogeneous_results} and \ref{subfig:Discounted_minmax_homogeneous_results} illustrates the cumulative discounted reward and the average reward, respectively.
In both figures, the Whittle index policy consistently achieves the highest performance compare to other policies. The difference in performance between the index and myopic policies can be explained with two reasons. The immediate reward for the middle state $1$ is higher for the passive action than the active action while the immediate reward for the highest state $2$ is higher for the active action. At the same time the passive action with TPM $P^p$ (stationary probability vector $[0.7951,0.1558,0.0491]^T$) ``pulls heavily'' towards lower states, in contrast to the active action TPM $P^a$ (stationary probability vector $[0.0109,0.0453,0.9438]^T$) which ``pulls heavily'' towards higher states. This model structure causes the myopic policy to perform much worse than the index policy which look ahead instead of choosing greedily.

\begin{figure}[!htp]
\hspace{-0.2cm}
\begin{tabular}{cc}
\begin{subfigure}[t]{0.23\textwidth}
    \centering
    \includegraphics[width=\linewidth]{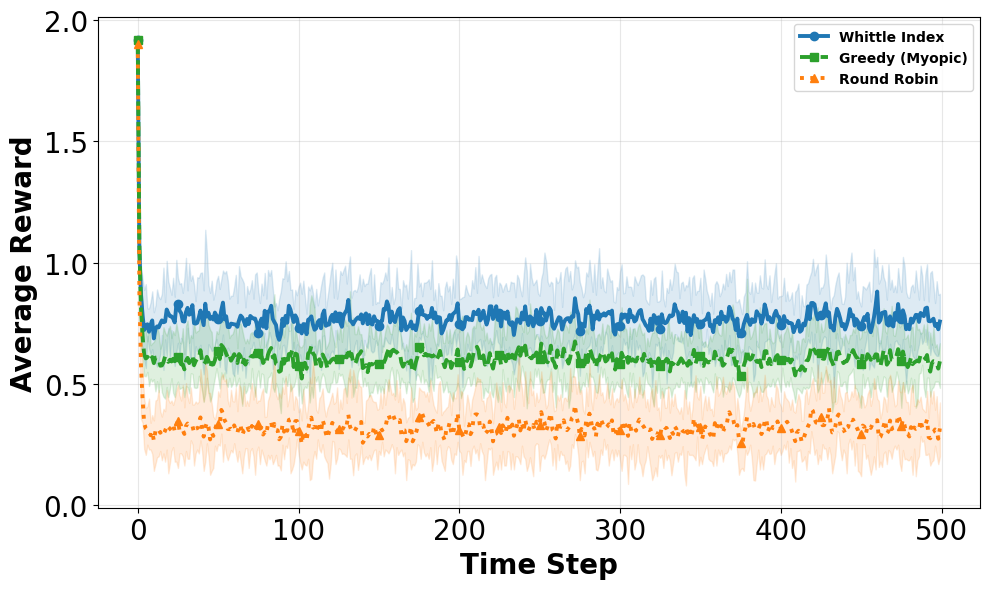}
    \caption{Average reward with min-max band}\label{subfig:Average_minmax_homogeneous_results}
\end{subfigure}
&
\hspace{-0.4cm}
\begin{subfigure}[t]{0.23\textwidth}
    \centering
    \includegraphics[width=\linewidth]{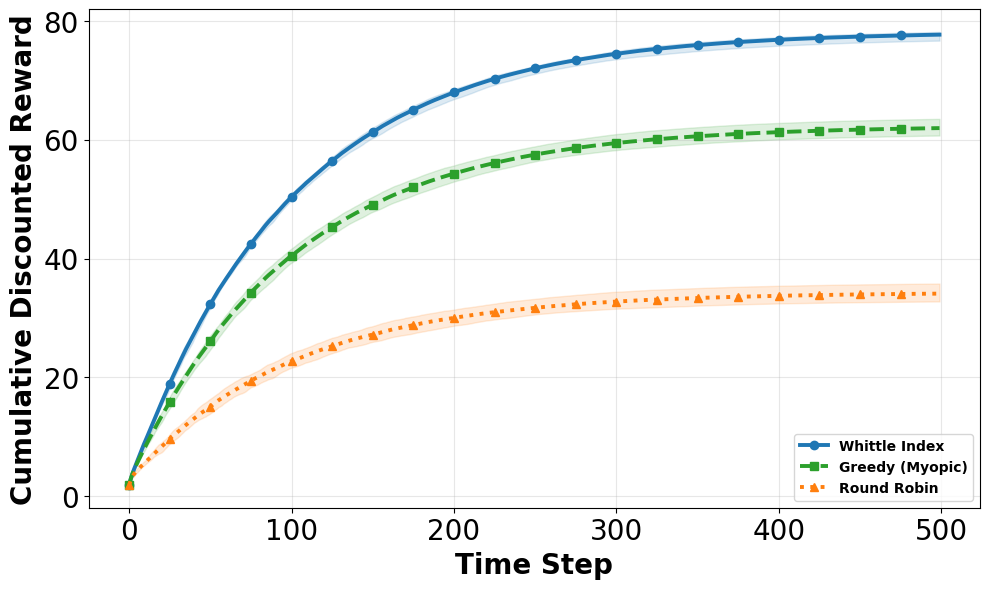}
    \caption{Discounted reward with min-max band}\label{subfig:Discounted_minmax_homogeneous_results}
\end{subfigure}
\end{tabular}
\caption{Comparison of policies for homogeneous arms: $N=50,$ $B=10$}
\label{fig:homogeneous-arms}
\end{figure}



\subsubsection{Heterogeneous arms} The system
consists of $N=50$ heterogeneous arms, where each arm is modeled as a
three-state Markov process with state space $\{0,1,2\}$. At each time step, the
scheduler selects $B=10$ arms for activation.


The passive transition matrix for each arm is structured as $P_n^p =
\begin{bmatrix}
1 & 0 & 0\\
0 & x_n & 1-x_n\\
0 & y_n & 1-y_n
\end{bmatrix},$
where
$
x_n \sim \mathcal{U}(0.6,0.9)$ and
$y_n \sim \mathcal{U}(0.6,x_n).
$

The active transition matrix is generated as $P_n^a =
\begin{bmatrix}
1-x_n' & x_n' & 0\\
0 & 1-z_{1,n} & z_{1,n}\\
0 & 1-z_{2,n} & z_{2,n}
\end{bmatrix},$ where the parameters are sampled such that the resulting transition
matrix satisfies the stochastic dominance condition with respect to the
passive transition matrix.


For each arm, a base reward parameter is sampled as
    $v_n \sim \mathcal{U}(0.5,2.0).$
The passive and active reward vectors are generated as $R_n^p =
[0,\;0.8v_n,\;1.6v_n] + \epsilon_n^p$ and $R_n^a =
[0,\;1.2v_n,\;2.4v_n] + \epsilon_n^a,$ where $\epsilon_n^p$ and $\epsilon_n^a$ denote small random perturbations
sampled independently from $\mathcal{U}(-0.1,0.1).$ 



 In Figures \ref{subfig:Average_minmax_heterogeneous_results} and \ref{subfig:Discounted_minmax_heterogeneous_results}, we observe that the Whittle index policy achieves performance close to the greedy policy. 
 
 Although transition probabilities differ between arms, the variations are structured and do not create substantial differences in their long-term activation values. Meanwhile, the heterogeneity of the rewards is more pronounced, making instantaneous rewards a strong indicator of arm quality. As a result, greedy policy often selects arms similar to those selected by the Whittle index policy. 
 


    

\begin{figure}[!htp]
\centering
 \begin{tabular}{cc}
\hspace{-.2cm}
\begin{subfigure}[t]{0.23\textwidth}
    \centering
    \includegraphics[width=\linewidth]{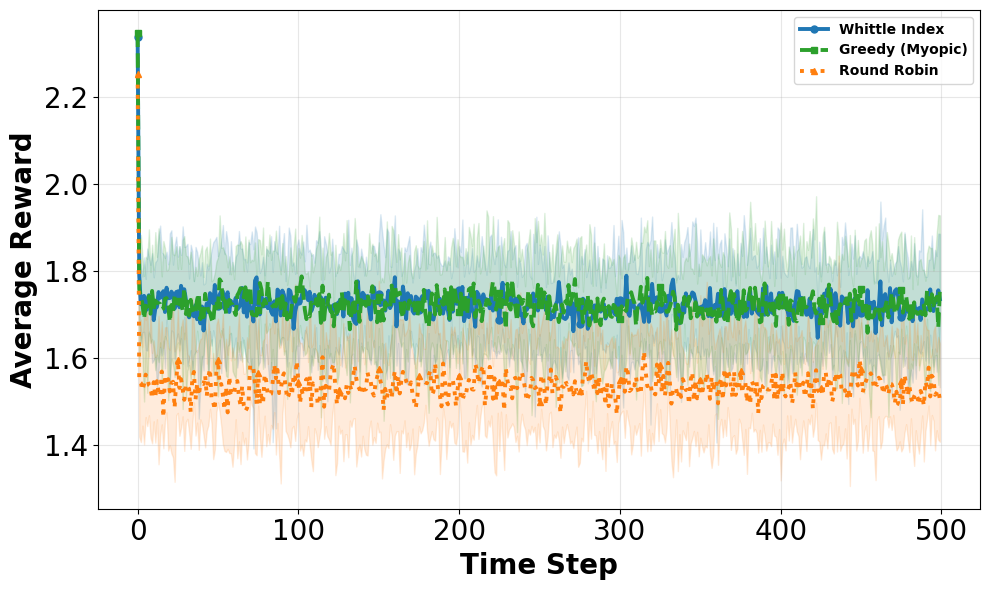}
    \caption{Average reward with min-max band}\label{subfig:Average_minmax_heterogeneous_results}
\end{subfigure}
&
\hspace{-.4cm}
\begin{subfigure}[t]{0.23\textwidth}
    \centering
    \includegraphics[width=\linewidth]{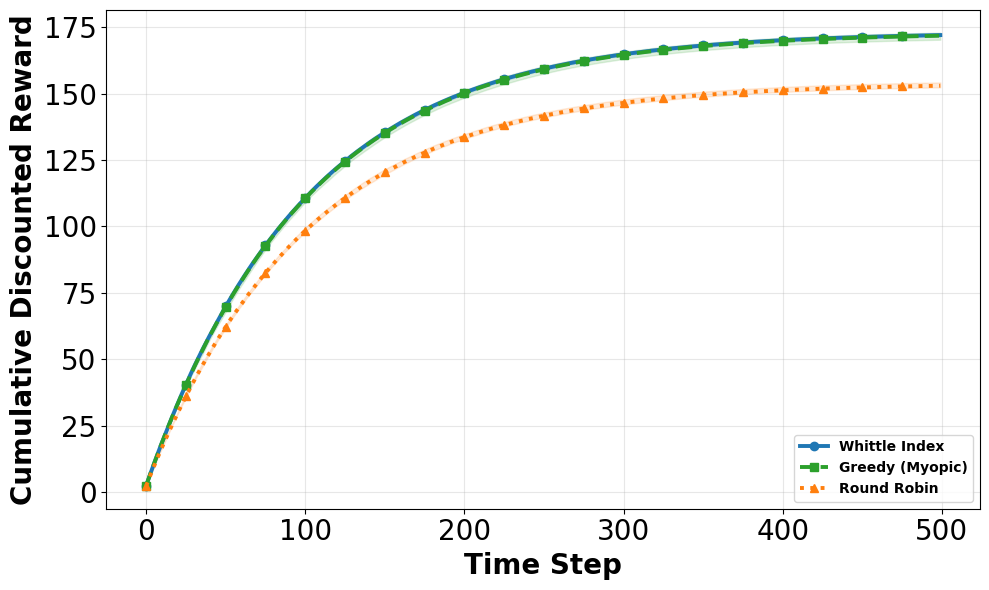}
    \caption{Discounted reward with min-max band}\label{subfig:Discounted_minmax_heterogeneous_results}
\end{subfigure}
\end{tabular}
\caption{Comparison of policies for heterogeneous arms. $N=50.$}
\label{fig:homogeneous_results}
\end{figure}
\subsubsection{Impact of reward structure on performance with heterogeneous arms }


We next study the effect of the reward structure, we retain the same heterogeneous transition-matrix discussed earlier but use a common reward structure across all arms. In the first example, we use $R^{\mathrm{a}} = [0,1,2],$ and $R^{\mathrm{p}} = [0,0.8,1.8].$  This is illustrated in Figure~\ref{fig:reward-structure-example1}, we observe that Whittle index policy and greedy policy is very close and maximizing the immediate reward aligns closely with maximizing the long-term reward.

\begin{figure}[!htp]
\hspace{-0.2cm}
\begin{tabular}{cc}
\begin{subfigure}[t]{0.23\textwidth}
    \centering
    \includegraphics[width=\linewidth]{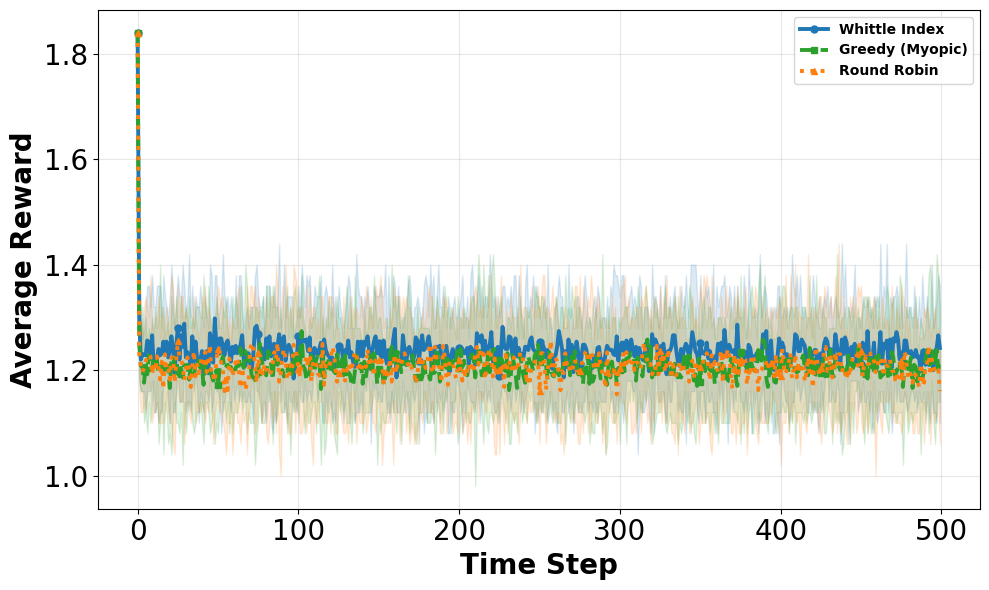}
    \caption{Average reward with min-max band}\label{subfig:Rewards2_Average_minmax_homogeneous_results}
\end{subfigure}
&
\hspace{-0.4cm}
\begin{subfigure}[t]{0.23\textwidth}
    \centering
    \includegraphics[width=\linewidth]{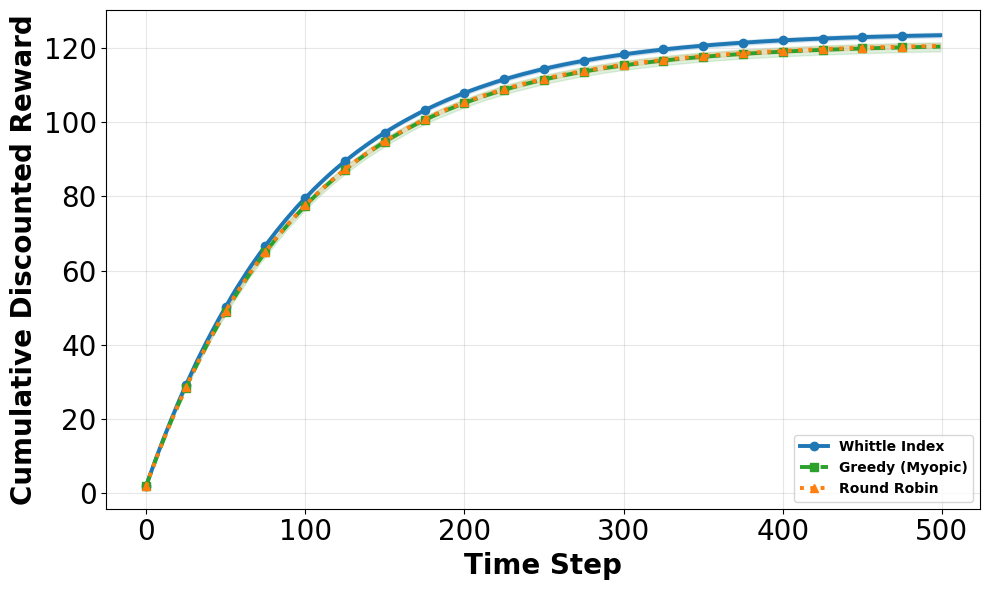}
    \caption{Discounted reward with min-max band}\label{subfig:Rewards2_Discounted_minmax_homogeneous_results}
\end{subfigure}
\end{tabular}
\caption{Heterogenous transition probability but identical rewards $R^{\mathrm{a}} = [0,1,2],$ and $R^{\mathrm{p}} = [0,0.8,1.8].$}
\label{fig:reward-structure-example1}
\end{figure}

In the next example, we consider 
rewards as follows.$R^{\mathrm{a}} = [0,1,9]$ and 
$R^{\mathrm{p}} = [0,0.8,8.8].$
The performance comparison is shown in Figure~\ref{fig:reward-structure-exampl2}


Our experiments reveal that the reward structure significantly influences the relative performance of the scheduling policies. 
The Whittle index policy consistently improves its performance relative to both the Greedy (myopic) and Round Robin policies.

\begin{figure}[!htp]
\hspace{-0.2cm}
\begin{tabular}{cc}
\begin{subfigure}[t]{0.23\textwidth}
    \centering
    \includegraphics[width=\linewidth]{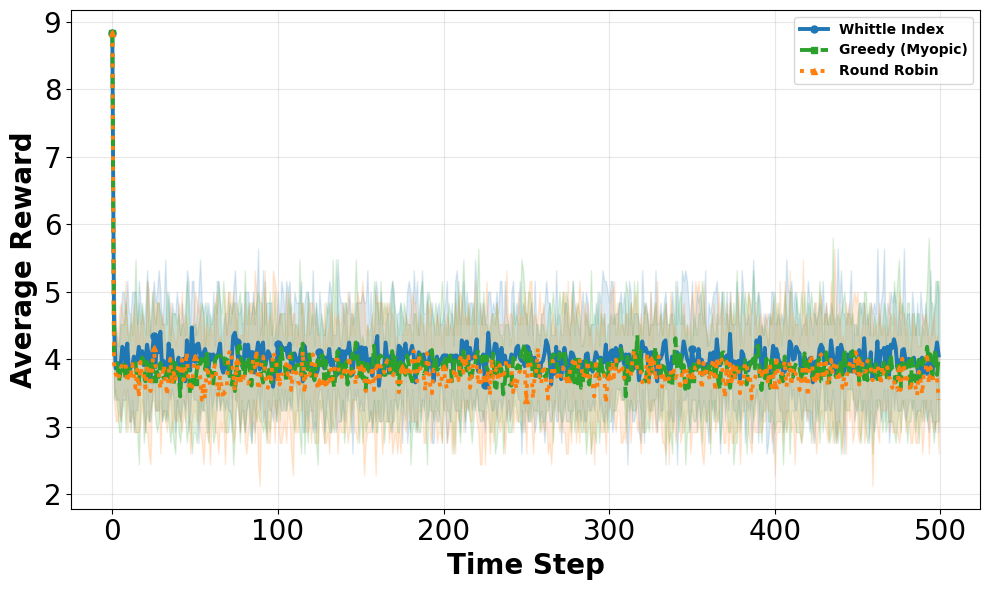}
    \caption{Average reward with min-max band}\label{subfig:Rewards9_Average_minmax_heterogeneous_results}
\end{subfigure}
&
\hspace{-.4cm}
\begin{subfigure}[t]{0.23\textwidth}
    \centering
    \includegraphics[width=\linewidth]{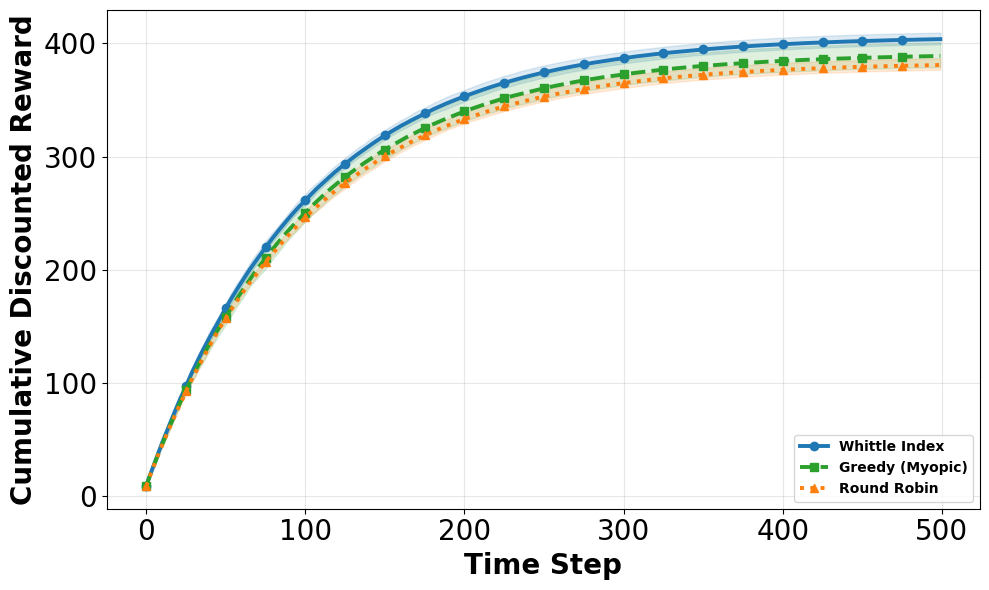}
    \caption{Discounted reward with min-max band}\label{subfig:Rewards9_Discounted_minmax_heterogeneous_results}
\end{subfigure}
\end{tabular}
\caption{Comparison of policies for heterogeneous transition probability matrices but identical rewards: $R^{\mathrm{a}} = [0,1,9]$ and 
$R^{\mathrm{p}} = [0,0.8,8.8].$ }
\label{fig:reward-structure-exampl2}
\end{figure}



\section{Conclusion}
\label{sec:conclusion}
We have studied three-state partially observable restless multi-armed bandits under the condition that the certainty about the state of an arm collapses when the arm is activated. It is referred to as three-state collapsing bandits. We showed that the optimal policy has a threshold type structure in the belief space  and present sufficiency condition using Lipschtiz bound on the value function.
We also proved the indexability for each bandit and provide sufficiency conditions under weaker model assumptions. We studied an index computation algorithms. Our simulation experiments demonstrate that the index policy is  better than the myopic and round-robin policies under some conditions and at least as good as the myopic policy in others examples. 

Future work will investigate generalized multi-state partially observable RMABs and their applications to public health intervention planning. We also plan to extend the analysis to average-reward criteria and to multi-state, multi-action partially observable RMABs.

\section{Acknowledgment}

Authors have used AI models in coding, editing  and proofreading of the paper. Authors take all responsibility for any error remains in the paper.




\bibliographystyle{IEEE}
\bibliography{restlessbandits}
\appendices

\section{Proof of Lemma~\ref{lemma:val-fun-decouple}}
\label{Proof-lemma:val-fun-decouple-appendix}

Denote the expressions on the right hand side of \eqref{eqn:DP-1} and  \eqref{eqn:DP-2} as $\mathcal{E}_1$ and $\mathcal{E}_2,$  respectively. We need to show that 
$\mathcal{E}_1(\mathcal{E}_2) = \mathcal{E}_2.$ That means, it suffices to show that the following expression $\mathcal{E}_1(\mathcal{E}_2)- \mathcal{E}_2 = 0.$ 

 Rewriting the Bellman operator,

\begin{align}
    \mathcal E_1(V)(\Omega)
=&
\max_{\bm a\in\{0,1\}^N}
\Biggl\{
\sum_{n=1}^N R_n(\bm\omega_n,a_n)
+\lambda\left(B-\sum_{n=1}^N a_n\right)\nonumber\\
&+\beta
\sum_{\bm o\in\mathcal O}
V(\tau(\bm\omega,\bm o,\bm a))
\prod_{n=1}^N
\Pr(o_n\mid\bm\omega_n,a_n)
\Biggr\}.\label{eq: E1}
\end{align}
Also,
\begin{align}
    \mathcal E_2(\Omega)
=
\sum_{n=1}^N V_n^\lambda(\bm\omega_n)
+
\frac{(B-N)\lambda}{1-\beta}.\label{eq: E2}
\end{align}

Substitute \ref{eq: E2} into \ref{eq: E1},

\begin{align}
\mathcal E_1(\mathcal E_2)(\Omega)
&=
\max_{\bm a}
\Bigg\{
\sum_{n=1}^N R_n(\bm\omega_n,a_n)
+\lambda\left(B-\sum_{n=1}^N a_n\right)
\nonumber\\
&+
\beta
\sum_{\bm o}
\left[
\sum_{n=1}^N
V_n^\lambda
\bigl(\tau_n(\bm\omega_n,o_n,a_n)\bigr)
+
\frac{(B-N)\lambda}{1-\beta}
\right]\nonumber\\
&\prod_{n=1}^N
\Pr(o_n\mid\bm\omega_n,a_n)
\Bigg\}. \label{Eq: E1E2}
\end{align}

Since $\sum_{\bm o\in\mathcal O}
\prod_{n=1}^N
\Pr(o_n\mid\bm\omega_n,a_n)=1,$
the constant term becomes
$\beta\frac{B\lambda}{1-\beta}.$ 
\remove{Therefore \ref{Eq: E1E2} becomes,
\begin{align}
\mathcal E_1&(\mathcal E_2)(\Omega)
=
\max_{\bm a}
\Bigg\{
\sum_{n=1}^N R_n(\bm\omega_n,a_n)
+\lambda B
-\lambda\sum_{n=1}^N a_n
\nonumber\\
&+
\beta\sum_{\bm o}
\sum_{n=1}^N
V_n^\lambda
(\tau_n(\bm\omega_n,o_n,a_n))
\prod_{n=1}^N
\Pr(o_n\mid\bm\omega_n,a_n)
\nonumber\\
&+
\frac{\beta\lambda(B-N)}{1-\beta}
\Bigg\}.\label{Eq: E1E21}
\end{align}
}
Since the observations are conditionally independent across arms,  \ref{Eq: E1E2} becomes,
\remove{
\begin{align*}
    &\sum_{\bm o} \sum_{n=1}^N V_n^\lambda (\tau_n(\bm\omega_n,o_n,a_n)) \prod_{m=1}^N \Pr(o_m\mid\bm\omega_m,a_m) \nonumber\\ &= \sum_{n=1}^N \sum_{o_n\in\mathcal O_n} V_n^\lambda (\tau_n(\bm\omega_n,o_n,a_n)) \Pr(o_n\mid\bm\omega_n,a_n). 
\end{align*}
}

 \begin{align} 
 \mathcal E_1(\mathcal E_2)(\Omega) &= \max_{\bm a} \Bigg\{ \sum_{n=1}^N \left[ R_n(\bm\omega_n,a_n)-\lambda a_n \right] \nonumber\\ &+ \beta\sum_{n=1}^N \sum_{o_n} V_n^\lambda (\tau_n(\bm\omega_n,o_n,a_n)) \Pr(o_n\mid\bm\omega_n,a_n) \nonumber\\ &+ \frac{\beta \lambda(B-\beta N)}{1-\beta} \Bigg\}\nonumber\\
  &=\frac{\lambda(B-\beta N)}{1-\beta} + \max_{\bm a} \sum_{n=1}^N \Bigg\{ R_n(\bm\omega_n,a_n) -\lambda a_n \nonumber\\ &+ \beta\sum_{o_n} V_n^\lambda (\tau_n(\bm\omega_n,o_n,a_n)) \Pr(o_n\mid\bm\omega_n,a_n) \Bigg\}.  \label{Eq: E1E23}
 \end{align} 

Single-arm Bellman equation is

\begin{align}
    V_n^\lambda(\bm\omega_n)& = \max_{a_n} \Biggl\{ R_n(\bm\omega_n,a_n) +\lambda(1-a_n) \nonumber\\&+ \beta\sum_{o_n} V_n^\lambda(\tau_n(\bm\omega_n,o_n,a_n)) \Pr(o_n\mid\bm\omega_n,a_n) \Biggr\}.
\end{align}
Hence
\begin{align}
  \mathcal E_1(\mathcal E_2)(\Omega) &= \frac{\lambda(B- \beta N)}{1-\beta} + \sum_{n=1}^N \left[ V_n^\lambda(\bm\omega_n)-\lambda \right]\nonumber\\
  &=
\sum_{n=1}^N V_n^\lambda(\bm\omega_n)
+
\frac{(B-N)\lambda}{1-\beta}.\label{Eq: E1E24}
\end{align}

Subtracting \ref{eq: E2} from \ref{Eq: E1E24},

 \begin{align*} \mathcal E_1(\mathcal E_2)&-\mathcal E_2 = 0. \end{align*} 
 \qed 

\section{Proof of Proposition 1}
\label{proof of Proposition 1}
 For $a=0$, $\tau_0(\bm{\omega})
=
(P^0)^\top\bm{\omega}.$
Therefore, $c_0
=
\delta(P^0),$
where $\delta(P^0)
=
\frac12
\max_{i,j}
\sum_{k=0}^{2}
|p_{ik}^0-p_{jk}^0|$ is the Dobrushin contraction coefficient.

For $a=1$, in the collapsing-bandit model, the continuation
value is $C_1(\bm{\omega})
=
\sum_{i=0}^{2}
\omega_i
V^\lambda(\Gamma_i).$

The effective contraction coefficient associated with this affine
continuation operator is $c_1
=
\frac12
\max_{i,j}
\|
\Gamma_i-\Gamma_j
\|_1.$

Thus, $c
=
\max_{a\in\{0,1\}}c_a.$ 
\begin{equation*}
\label{eq:c-collapsing}
c
=
\max
\left\{
\delta(P^0),
\frac12
\max_{i,j}
\|
\Gamma_i-\Gamma_j
\|_1
\right\}.
\end{equation*} 

Both terms are at most one $0\le c\le1.$

\qed 

\section{Proof of Lemma $4$ (Lipschitz bound on value function)} \label{Appendix: Lemma 4}

We need to show that the value function is Lipschitz in the belief, 
$|V^\lambda(\omega)-V^\lambda(\omega')|
\le
\frac{L_R}{1-\beta c}\|\omega-\omega'\|_1,$ where
$L_R=\frac12\max_{a\in\{0,1\}}\operatorname{span}(r^a)$ and $c=\max\left\{\delta(P^0),\frac12D_\Gamma\right\},
$ $
D_\Gamma=\max_{i,j}\|\Gamma_i-\Gamma_j\|_1.$ We prove the result by mathematical induction.

For $a \in \{0,1\}$, the expected one-stage reward is linear in the belief: $R(\omega,a)=\sum_i\omega_i r(i,a).$
For any two probability vectors $\omega,\omega',$ the standard total-variation bound gives $|R(\omega,a)-R(\omega',a)|
\le
\frac12\operatorname{span}(r^a)\|\omega-\omega'\|_1
\le
L_R\|\omega-\omega'\|_1.$

The subsidy $\lambda$ does not affect this bound because it is a constant with respect to the belief. Now, initialize $V_0^\lambda=0,$ then $|V_1^{\lambda}(\bm{\omega})- V_1^{\lambda}(\bm{\omega}')|\leq L_R\|\omega-\omega'\|_1.$ Then $V_{n+1}^\lambda$ can be obtained from $V_n^\lambda$ through the Bellman operator. Suppose that $V_n^\lambda$ is $L_n$-Lipschitz, so that $|V_n^\lambda(\omega)-V_n^\lambda(\omega')|
\le
L_n\|\omega-\omega'\|_1.$

For $a=0$, belief transition is $\tau(\omega)=(P^0)^\top\omega.$ By the Dobrushin contraction inequality, $\|\tau(\omega)-\tau(\omega')\|_1
\le
\delta(P^0)\|\omega-\omega'\|_1.$ Therefore, $|V_n^\lambda(\tau(\omega))
-
V_n^\lambda(\tau(\omega'))|
\le
L_n\delta(P^0)\|\omega-\omega'\|_1.$ Thus, the passive continuation term has Lipschitz modulus $\delta(P^0)L_n.$ 

For the active action, the state is revealed. Hence its continuation value is $C_1V_n^\lambda(\omega)
=
\sum_i\omega_iV_n^\lambda(\Gamma_i).$ The difference between this quantity at two beliefs is
$\sum_i(\omega_i-\omega_i')V_n^\lambda(\Gamma_i).$

Applying the total-variation bound to the three numbers
 $V_n^\lambda(\Gamma_0),V_n^\lambda(\Gamma_1),V_n^\lambda(\Gamma_2),$ we obtain $|C_1V_n^\lambda(\omega)-C_1V_n^\lambda(\omega')|
\le
\frac12
\operatorname{span}\{V_n^\lambda(\Gamma_0),
V_n^\lambda(\Gamma_1),
V_n^\lambda(\Gamma_2)\}
\|\omega-\omega'\|_1.$

By the induction hypothesis, $|V_n^\lambda(\Gamma_i)-V_n^\lambda(\Gamma_j)|
\le
L_n\|\Gamma_i-\Gamma_j\|_1.$ 

Consequently, $\operatorname{span}\{V_n^\lambda(\Gamma_0),
V_n^\lambda(\Gamma_1),
V_n^\lambda(\Gamma_2)\}
\le
L_nD_\Gamma.$

Hence, $|C_1V_n^\lambda(\omega)-C_1V_n^\lambda(\omega')|
\le
\frac12L_nD_\Gamma
\|\omega-\omega'\|_1.$ Thus, the active continuation operator has Lipschitz modulus $D_\Gamma/2.$

Define, $c_0=\delta(P^0)$, $c_1=\frac12D_\Gamma,$ and let $c=\max\{c_0,c_1\}.$
For either action, the difference between the corresponding Q-values is therefore bounded by the difference in immediate rewards plus the difference in continuation values. Hence,

$|Q_{n+1}^\lambda(\omega,a)
-
Q_{n+1}^\lambda(\omega',a)|
\le
(L_R+\beta c_aL_n)\|\omega-\omega'\|_1.$

The Bellman value is the maximum of the two Q-values. Since the maximum operator is nonexpansive,
\begin{align*}
    |V_{n+1}^\lambda(\omega)-V_{n+1}^\lambda(\omega')|
\le
\max_a
|Q_{n+1}^\lambda(\omega,a)
-
Q_{n+1}^\lambda(\omega',a)|.
\end{align*}

Therefore, 
\begin{align*}
    |V_{n+1}^\lambda(\omega)-V_{n+1}^\lambda(\omega')|
\le
(L_R+\beta cL_n)\|\omega-\omega'\|_1.
\end{align*}

Thus a valid Lipschitz modulus satisfies $L_{n+1}\le L_R+\beta cL_n.$ Since $L_0=0$, repeated substitution gives $L_n
\le
L_R\sum_{k=0}^{n-1}(\beta c)^k
=
L_R\frac{1-(\beta c)^n}{1-\beta c}.$

Finally, $ \delta(P^0)\le1$, while $\|\Gamma_i-\Gamma_j\|_1\le2 $ because $\Gamma_i $and $\Gamma_j $ are probability distributions. Therefore $ 0\le c\le1$. Since $0<\beta<1$, we have $\beta c<1.$

The discounted Bellman operator is a contraction, so $V_n^\lambda$ converges uniformly to $V^\lambda.$ Taking the limit as $n \rightarrow \infty $ gives $|V^\lambda(\omega)-V^\lambda(\omega')|
\le
\frac{L_R}{1-\beta c}
\|\omega-\omega'\|_1.$

\qed



\begin{proposition}[Bound on active continuation term]
For $\bm{\omega}, \bm{\omega}' \in \Pi,$ we obtain 
\begin{equation*}
\label{eq:active-bound}
\left|
\sum_{i=0}^{2}
(\omega^i-\omega^{i,'})
V^\lambda(\Gamma_i)
\right|
\le
\frac12
\frac{L_R}{1-\beta c}
D_\Gamma
\|
\bm{\omega}
-
\bm{\omega}'
\|_1.
\end{equation*}
\end{proposition}

\textbf{Proof:} Consider the active continuation increment
$C_1(\bm{\omega})
=
\sum_{i=0}^{2}
\omega_i
V^\lambda(\Gamma_i).$
Then, $C(\bm{\omega})
-
C(\bm{\omega}')
=
\sum_{i=0}^{2}
(\omega^i-\omega^{i,'})
V^\lambda(\Gamma_i).$

Next we have 
\begin{align*}
&\left|
\sum_{i=0}^{2}
(\omega^i-\omega^{i,'})
V^\lambda(\Gamma_i)
\right|\\
&\le
\frac12
\operatorname{span}
\left(
V^\lambda(\Gamma_0),
V^\lambda(\Gamma_1),
V^\lambda(\Gamma_2)
\right)
\|
\bm{\omega}
-
\bm{\omega}'
\|_1.
\end{align*}
Furthermore, $|V^\lambda(\Gamma_i)-V^\lambda(\Gamma_j)|
\le
\frac{L_R}{1-\beta c}
\|
\Gamma_i-\Gamma_j
\|_1.$ Here, $\operatorname{span}
\left(
V^\lambda(\Gamma_0),
V^\lambda(\Gamma_1),
V^\lambda(\Gamma_2)
\right)
\le
\frac{L_R}{1-\beta c}
D_\Gamma,$
where $D_\Gamma
:=
\max_{i,j}
\|
\Gamma_i-\Gamma_j
\|_1.$

Therefore, $\left|
\sum_{i=0}^{2}
(\omega_i-\omega_i')
V^\lambda(\Gamma_i)
\right|
\le
\frac12
\frac{L_R}{1-\beta c}
D_\Gamma
\|
\bm{\omega}
-
\bm{\omega}'
\|_1.$

\qed

\begin{proposition}[Lower bound on the difference of passive advantages]
\label{prop:advantage-difference-lower-bound}

Assume that \(V^\lambda\) is Lipschitz continuous in the
\(\ell_1\)-norm with  Lipschitz constant $L_V=\frac{L_R}{1-\beta c}$. $V^\lambda$ is non-decreasing with respect to the MLR order and that the passive belief update $\tau$ preserves the MLR order.

 Then, for every pair of beliefs
\(\bm{\omega},\bm{\omega}'\),
\begin{equation*}
\label{eq:advantage-difference-final-bound}
\Delta(\bm{\omega},\lambda)
-
\Delta(\bm{\omega}',\lambda)
\geq
d(\bm{\omega})
-
d(\bm{\omega}')
-
\frac{\beta L_VD_\Gamma}{2}
\left\|
\bm{\omega}
-
\bm{\omega}'
\right\|_1.
\end{equation*}

\end{proposition}

\textbf{Proof:}

Proof of proposition is given here. 

Subtracting the two passive advantages, the subsidy terms cancel and we
obtain
\begin{align*}
\Delta(\bm{\omega},\lambda)
-
\Delta(\bm{\omega}',\lambda)
&=
d(\bm{\omega})
-
d(\bm{\omega}')\\
&+
\beta
\left[
V^\lambda(\tau(\bm{\omega}))
-
V^\lambda(\tau(\bm{\omega}'))
\right]\\
&-
\beta
\sum_{i=0}^{2}
(\bm{\omega}^i-\bm{\omega}^{i,'})
V^\lambda(\Gamma_i).
\label{eq:advantage-difference-expansion}
\end{align*}
\remove{
By the inequality \(x\geq-|x|\), the Lipschitz continuity of
\(V^\lambda\), and the contraction property of \(\tau\),
\begin{align*}
V^\lambda(\tau(\bm{\omega}))
-
V^\lambda(\tau(\bm{\omega}'))
& \geq
-
\left|
V^\lambda(\tau(\bm{\omega}))
-
V^\lambda(\tau(\bm{\omega}'))
\right|\\
&\geq
-
L_V
\left\|
\tau(\bm{\omega})
-
\tau(\bm{\omega}')
\right\|_1\\
&\geq
-
L_Vc_0
\left\|
\bm{\omega}
-
\bm{\omega}'
\right\|_1.
\end{align*}
}
Next, we define $\bm v_\Gamma^\lambda
=
\left(
V^\lambda(\Gamma_0),
V^\lambda(\Gamma_1),
V^\lambda(\Gamma_2)
\right)^\top.$

We also have 
\begin{align*}
\left|
\sum_{i=0}^{2}
(\omega^i-\omega^{i,'})
V^\lambda(\Gamma_i)
\right|
\leq
\frac{1}{2}
\operatorname{span}
\left(
\bm v_\Gamma^\lambda
\right)
\left\|
\bm{\omega}
-
\bm{\omega}'
\right\|_1.
\end{align*}

Moreover, we obtain 
\begin{align*}
\operatorname{span}
\left(
\bm v_\Gamma^\lambda
\right)
&=
\max_{i,j}
\left|
V^\lambda(\Gamma_i)
-
V^\lambda(\Gamma_j)
\right|
\nonumber
\leq
L_V
\max_{i,j}
\|\Gamma_i-\Gamma_j\|_1
\nonumber\\
&=
L_VD_\Gamma.
\end{align*}
Hence, $-
\sum_{i=0}^{2}
(\omega_i-\omega_i')
V^\lambda(\Gamma_i)
\geq
-
\frac{L_VD_\Gamma}{2}
\left\|
\bm{\omega}
-
\bm{\omega}'
\right\|_1.$

Since $\tau(\omega)$ preserve MLR order, by lemma \ref{lemma:monotonicity},
$V^\lambda(\tau(\bm{\omega}))
-
V^\lambda(\tau(\bm{\omega}'))\geq0.$

Therefore,
\begin{align*}
    \Delta(\bm{\omega},\lambda)
-
\Delta(\bm{\omega}',\lambda)
\geq
d(\bm{\omega})
-
d(\bm{\omega}')
-
\frac{\beta L_VD_\Gamma}{2}
\left\|
\bm{\omega}
-
\bm{\omega}'
\right\|_1.
\end{align*}
\qed 

\begin{proposition}[Reward-difference bound under MLR ordering]
Define the statewise immediate passive action  advantage by $d_i
=
r(i,0)- r(i,1),
\quad
i\in\{0,1,2\}.$ By assumption $1$,  we have  $d_0
\leq
d_1
\leq
d_2.$ 
Define the increasing reward-difference margin
$m_d 
:=
\min
\{
d_1-d_0,
d_2-d_1\}$
Assume that $\bm{\omega}
\succeq_{\mathrm{MLR}}
\bm{\omega}'.
$
Then, $d(\bm{\omega})
-
d(\bm{\omega}')
\geq
\frac{m_d}{2}
\left\|
\bm{\omega}
-
\bm{\omega}'
\right\|_1.$
\label{prop:monotone-diff-reward-bound}
\end{proposition}

\textbf{Proof:} Consider $\bm{\omega}
\succeq_{\mathrm{MLR}}
\bm{\omega}',$
and the state space is linearly ordered as $0<1<2,$ 
MLR dominance implies first-order stochastic dominance: $\bm{\omega}
\succeq_{\mathrm{FOSD}}
\bm{\omega}'.$
For three states, this means $\bm{\omega}^0
\leq
\bm{\omega}^{0,'}, $
and 
$\bm{\omega}^0+\bm{\omega}^1
\leq
\bm{\omega}^{0,'}+\bm{\omega}^{1,'}.$
Equivalently, $\bm{\omega}^1+\bm{\omega}^2
\geq
\bm{\omega}^{1,'}+\bm{\omega}^{2,'},$
 and $\bm{\omega}^2
\geq
\bm{\omega}^{2,'}.$

We can rewrite  
\begin{align*}
d(\bm{\omega})
&=
\bm{\omega}^0 d_0
+
\bm{\omega}^1 d_1
+
\bm{\omega}^2 d_2
\nonumber\\
&=
d_0
+
(d_1-d_0)
(\bm{\omega}^1+\bm{\omega}^2)
+
(d_2-d_1)\bm{\omega}^2.
\end{align*}

Similarly,
\begin{align*}
d(\bm{\omega}')
&=
d_0
+
(d_1-d_0)
(\bm{\omega}^{1'}+\bm{\omega}^{2'})
+
(d_2-d_1)\bm{\omega}^{2'}.
\end{align*}

Subtracting gives
\begin{align*}
d(\bm{\omega})
-
d(\bm{\omega}')
&=
(d_1-d_0)
\left[
(\bm{\omega}^1+\bm{\omega}^2)
-
(\bm{\omega}^{1'}+\bm{\omega}^{2'})
\right]\\
&+
(d_2-d_1)
(\bm{\omega}^2-\bm{\omega}^{2'}).
\end{align*}
Note that both terms in the preceding equation are nonnegative. Since $d_1-d_0\geq m_d,$ and $d_2-d_1\geq m_d,$
we obtain
\begin{align*}
d(\bm{\omega})
-
d(\bm{\omega}')
\geq
m_d
\Big[
(\bm{\omega}^1-(\bm{\omega}^{1'})+
+
2(\bm{\omega}^2-\bm{\omega}^{2'})
\Big].
\end{align*}

Define $x_i
=
\bm{\omega}^i-\bm{\omega}^{i,'}$ for $
i\in\{0,1,2\}.$

Since both beliefs are probability vectors $x_0+x_1+x_2=0.$
The FOSD inequalities imply $x_0\leq0,$ and $x_2\geq0.$
Moreover,
\begin{align*}
(\bm{\omega}^1+\bm{\omega}^2)
-
(\bm{\omega}^{1'}+\bm{\omega}^{2'})
+
\bm{\omega}^2-\bm{\omega}^{2'}
=
-x_0+x_2.
\end{align*}

Also, $x_1=-x_0-x_2,$
we have $|x_1|
\leq
-x_0+x_2.$

Therefore,
\begin{align*}
\left\|
\bm{\omega}
-
\bm{\omega}'
\right\|_1
&=
|x_0|+|x_1|+|x_2|
=
-x_0+|x_1|+x_2
\nonumber\\
&\leq
2(-x_0+x_2).
\end{align*}

Hence, $-x_0+x_2
\geq
\frac12
\left\|
\bm{\omega}
-
\bm{\omega}'
\right\|_1.$

Thus, we obtain $d(\bm{\omega})
-
d(\bm{\omega}')
\geq
\frac{m_d}{2}
\left\|
\bm{\omega}
-
\bm{\omega}'
\right\|_1.$

\section{Proof of Lemma 5}
\label{Appendix:Lemma 5}
Since $d_0\leq d_1\leq d_2,$
the function $i\mapsto d_i$ is nondecreasing in the state $i$.
Moreover, MLR dominance implies first-order stochastic
dominance. Therefore, if $\bm{\omega}\geq_r\bm{\omega}',$
then the expectation of every nondecreasing function of the
state is larger under $\bm{\omega}$ than under
$\bm{\omega}'$. In particular,
\[
\sum_{i=0}^{2}\omega^i d_i
\geq
\sum_{i=0}^{2}{\omega'}^i d_i.
\]
Hence, $d(\bm{\omega})\geq d(\bm{\omega}').$


\section{Theorem 1 (Sufficiency conditions for threshold policy)} \label{Appendix: Theorem 1}

We provide the proof of theorem in two steps:  
\begin{enumerate}
\item We show that if  $d(\bm{\omega})
-
d(\bm{\omega}')
\ge
\frac{\beta L_R D_\Gamma}{2(1-\beta c)}
\|
\bm{\omega}
-
\bm{\omega}'
\|_1,$ then $\Delta(\bm{\omega},\lambda)
\ge
\Delta(\bm{\omega}',\lambda).$ This follows from Proposition~\ref{prop:advantage-difference-lower-bound}.
This is further improved in next step. 
\item   If $m_d
\geq
\frac{\beta L_R D_\Gamma}
{1-\beta c_1},$ then passive action advantage
$\Delta(\bm{\omega}, \lambda )
=
Q(\bm{\omega},0)-Q(\bm{\omega},1)$
is increasing with belief $\bm{\omega}.$

From Proposition~\ref{prop:monotone-diff-reward-bound}, we have
\begin{equation*}
d(\bm{\omega})
-
d(\bm{\omega}')
\geq
\frac{m_d}{2}
\left\|
\bm{\omega}
-
\bm{\omega}'
\right\|_1.
\end{equation*}
Thus 
\begin{eqnarray*}
    \Delta(\bm{\omega},\lambda)
-
\Delta(\bm{\omega}',\lambda)
\geq \frac{m_d}{2}
\left\|
\bm{\omega}
-
\bm{\omega}'
\right\|_1  - 
\frac{\beta L_VD_\Gamma}{2}
\left\|
\bm{\omega}
-
\bm{\omega}'
\right\|_1.
\end{eqnarray*}
This can be written as follows. 
\begin{eqnarray*}\Delta(\bm{\omega},\lambda)
-
\Delta(\bm{\omega}',\lambda)
\geq \left[\frac{m_d}{2}
  - 
\frac{\beta L_VD_\Gamma}{2}
\right]
\left\|
\bm{\omega}
-
\bm{\omega}'
\right\|_1.
\end{eqnarray*}
We want  $\left[\frac{m_d}{2}
  - 
\frac{\beta L_VD_\Gamma}{2}
\right] \geq 0$ then $\Delta(\bm{\omega},\lambda)
-
\Delta(\bm{\omega}',\lambda) \geq 0.$
In other words,  
\begin{eqnarray*}
    m_d \geq \frac{\beta L_R D_{\Gamma}}{(1-\beta c)}. 
\end{eqnarray*}
Thus passive action advantage $\Delta(\bm{\omega},\lambda) $ is increasing in belief $\bm{\omega}$  for fixed $\lambda.$
This completes the proof. 
\end{enumerate}

\qed

\section{Proof of Theorem 2}
\label{Proof of Theorem 2}
Let $P^0$ and $P^1$ be the TP2 transition matrices of the
three-state birth-death model. We verify the sufficient condition
of the general threshold-policy theorem.

For $a\in\{0,1\}$, write the rows of $P^a$ as
$r_0^a,r_1^a,r_2^a.$ The Dobrushin coefficient is
$\delta(P^a)
=
1-\min_{i,j\in\{0,1,2\}}
\sum_{k=0}^2\min\{P^a_{ik},P^a_{jk}\}.$ 

Since $p_a+q_a\leq 1$, the overlap between rows $0$ and
$1$ is at least $q_a$, the overlap between rows $1$ and
$2$ is at least $p_a$, and the overlap between rows $0$
and $2$ is $\min\{p_a,q_a\}$.  Consequently, $\delta(P^a)\leq 1-\min\{p_a,q_a\}
\leq 1-m_P.$ In particular, $\delta(P^0)\leq 1-m_P.$ Next, for the active-action update, the rows of $P^1 $give
\begin{align*}
    \|r_0^1-r_2^1\|_1
&=(1-p_1)+|p_1-q_1|+(1-q_1)\\
&=2\bigl(1-\min\{p_1,q_1\}\bigr).
\end{align*}

The distances between rows $0$ and $1$, and between rows
$1$ and $2$, are no larger than this quantity. Hence
\begin{align*}
    D_\Gamma
:=\max_{i,j\in\{0,1,2\}}
\|\Gamma_i-\Gamma_j\|_1
=2\bigl(1-\min\{p_1,q_1\}\bigr)
\end{align*}
where \(\Gamma_i\) denotes the active-action update associated
with state \(i\).
Using the contraction parameter in the general theorem,
\[
c=\max\left\{\delta(P^0),\frac{D_\Gamma}{2}\right\},
\]
we obtain
\[
c
\leq
\max\left\{1-m_P,\,
1-\min\{p_1,q_1\}\right\}
=1-m_P,
\]
because \(\min\{p_1,q_1\}\geq m_P\).

The value-function Lipschitz bound therefore satisfies
\[
L_V=\frac{L_R}{1-\beta c}
\leq
\frac{L_R}{1-\beta(1-m_P)}.
\]
The general threshold theorem guarantees a threshold-type
optimal policy whenever
\[
m_d\geq
\frac{\beta L_R D_\Gamma}{1-\beta c}.
\]
Therefore,  $m_d
\geq
\frac{2 \beta\bigl(1-\min\{p_1,q_1\}\bigr)}
     {1-\beta(1-m_P)}\, L_R,$.

Therefore, \ref{thm:threshold-condition} implies the sufficient condition for an MLR-monotone optimal policy.

\section{Proof of Lemma~\ref{Lemma:Uniform-Lip-sensitivity}}
\label{Proof:Lemma-Uniform-Lip-sensitivity}


We study  sufficient conditions under which the subsidy
sensitivity $G_\lambda(\bm{\omega})
:=
\frac{\partial V^\lambda(\bm{\omega})}{\partial\lambda}$
is uniformly Lipschitz continuous with respect to the belief
state. We want show that there exists
 a finite constant $L_G\geq0$, independent of the subsidy
$\lambda$, such that
\begin{equation}
\label{eq:G-lipschitz}
\left|
G_\lambda(\bm{\omega})
-
G_\lambda(\bm{\omega}')
\right|
\leq
L_G\|\bm{\omega}-\bm{\omega}'\|_1,
\end{equation}
$\forall \bm{\omega},\bm{\omega}'\in\Pi$,$\forall
\lambda\in\mathbb R$ and where $L_G = \frac{\beta}{2(1-\beta)}.$

We derive the result using the differentiated Bellman equations,
the contraction of passive belief dynamics, and regularity
at optimal-action switching boundaries.

Recall the Bellman equation
\begin{eqnarray*}
    V^\lambda(\bm{\omega}) &=&
\max\left\{
Q^\lambda(\bm{\omega},0),
Q^\lambda(\bm{\omega},1)
\right\}, \\
Q^\lambda(\bm{\omega},0)
&=&
R(\bm{\omega},0)+\lambda
+\beta V^\lambda(\tau(\bm{\omega})),
\label{eq:Q0-G-lipschitz}\\
Q^\lambda(\bm{\omega},1)
&=&
R(\bm{\omega},1)
+\beta\sum_{s=0}^{2}
\omega_sV^\lambda(\Gamma_s).
\end{eqnarray*}
The passive belief transition is $\tau(\bm{\omega})=(P^0)^\top\bm{\omega}.$
Let $c_0:=\delta(P^0)
=
\frac12\max_{i,j}
\sum_{s=0}^{2}|P^0_{is}-P^0_{js}|$
denote the Dobrushin contraction coefficient.
Then
\begin{equation*}
\label{eq:passive-contraction-G}
\|\tau(\bm{\omega})-\tau(\bm{\omega}')\|_1
\leq
c_0\|\bm{\omega}-\bm{\omega}'\|_1.
\end{equation*}
Since $P^0$ is stochastic, we have$0\leq c_0\leq1,$ and $\beta c_0<1.$
Define the passive and active regions:
\[
\Pi_\lambda^0
=
\left\{
\bm{\omega}\in\Pi:
Q^\lambda(\bm{\omega},0)>
Q^\lambda(\bm{\omega},1)
\right\},
\]
and
\[
\Pi_\lambda^1
=
\left\{
\bm{\omega}\in\Pi:
Q^\lambda(\bm{\omega},1)>
Q^\lambda(\bm{\omega},0)
\right\}.
\]
The switching boundary is
\[
\mathcal T_\lambda
=
\left\{
\bm{\omega}\in\Pi:
Q^\lambda(\bm{\omega},0)
=
Q^\lambda(\bm{\omega},1)
\right\}.
\]
Suppose the subsidy derivative exists at the relevant
continuation beliefs.
On the passive region, differentiating
$V^\lambda(\bm{\omega})$ we get 
\begin{equation*}
\label{eq:G-passive-branch}
G_\lambda(\bm{\omega})
=
1+
\beta G_\lambda(\tau(\bm{\omega})).
\end{equation*}
On the active region, differentiating
$V^\lambda(\bm{\omega})$ we get 
\begin{equation}
\label{eq:G-active-branch}
G_\lambda(\bm{\omega})
=
\beta
\sum_{s=0}^{2}
\omega_sG_\lambda(\Gamma_s).
\end{equation}
These equations characterize the subsidy sensitivity
separately in the two optimal-action regions.

For an admissible policy $\phi$, define
\[
H_\phi(\bm{\omega})
=
\mathbb E_{\phi,\bm{\omega}}
\left[
\sum_{t=0}^{\infty}
\beta^t\mathbf 1_{\{A_t=0\}}
\right].
\]
Since $\mathbf 1_{\{A_t=0\}}\le1.$, we obtain
\[
0
\leq
H_\phi(\bm{\omega})
\leq
\sum_{t=0}^{\infty}\beta^t
=
\frac{1}{1-\beta}.
\]
Whenever the optimal value function is differentiable
with respect to $\lambda$, its derivative is the expected
discounted passive occupation time under a subsidy-optimal
policy:
\[
G_\lambda(\bm{\omega})
=
H_{\phi_\lambda^*}(\bm{\omega}).
\]
Therefore,
\begin{equation*}
\label{eq:G-uniform-bounded}
0
\leq
G_\lambda(\bm{\omega})
\leq
\frac{1}{1-\beta}.
\end{equation*}
This bound is uniform over the belief space and all subsidies.

\textbf{(a) Lipschitz Bound in the Active Region:}
Let $\bm{\omega},\bm{\omega}'\in\Pi.$
From \ref{eq:G-active-branch},
\begin{align}
G_\lambda(\bm{\omega})
-
G_\lambda(\bm{\omega}')
=
\beta
\sum_{s=0}^{2}
(\bm{\omega}^s-\bm{\omega}^{s'})G_\lambda(\Gamma_s).
\end{align}
Since $\bm{\omega}$ and $\bm{\omega}'$ are probability
distributions, $\sum_{s=0}^{2}(\bm{\omega}^s-\bm{\omega}^{s'})=0.$
For any real numbers $g_0,g_1,g_2$ and probability
vectors $\bm{\omega},\bm{\omega}'$, the total-variation
inequality gives
\[
\left|
\sum_{s=0}^{2}
(\bm{\omega}^s-\bm{\omega}^{s'})g_s
\right|
\leq
\frac12
\left(
\max_s g_s-\min_s g_s
\right)
\|\bm{\bm{\omega}}-\bm{\omega}'\|_1.
\]
Applying this inequality with
$g_s=G_\lambda(\Gamma_s)$ yields
\begin{align}
|G_\lambda(\bm{\omega})
-G_\lambda(\bm{\omega}')|
\leq
\frac{\beta}{2}
\operatorname{span}
\{G_\lambda(\Gamma_s)\}_{s=0}^{2}
\|\bm{\omega}-\bm{\omega}'\|_1.
\end{align}
Note that we have $\operatorname{span}
\{G_\lambda(\Gamma_s)\}_{s=0}^{2}
\leq
\frac{1}{1-\beta}.$
Hence 
\begin{equation}
\label{eq:G-active-Lipschitz}
|G_\lambda(\bm{\omega})
-G_\lambda(\bm{\omega'})|
\leq
\frac{\beta}{2(1-\beta)}
\|\bm{\omega}-\bm{\omega'}\|_1.
\end{equation}
Thus, the active branch has the uniform Lipschitz
constant
$L_G:=\frac{\beta}{2(1-\beta)}.$

\textbf{(b)Lipschitz bound in Passive Region:}
This is not very obvious. Here. we consider few cases.

\textbf{A) Repeated Passive-Action Recursion}
Suppose that two beliefs $\bm{\omega}$ and
$\bm{\omega}'$ remain in the passive region
for the first $n$ decision epochs: $\tau^k(\bm{\omega}),
\tau^k(\bm{\omega}')
\in\Pi_\lambda^0,$  $\forall k\in \{0,\ldots,n-1\}.$
Repeated application of the passive
sensitivity equation yields
\begin{eqnarray*}
G_\lambda(\bm{\omega})
=
\sum_{k=0}^{n-1}\beta^k
+
\beta^nG_\lambda(\tau^n(\bm{\omega})), \\
G_\lambda(\bm{\omega}')
=
\sum_{k=0}^{n-1}\beta^k
+
\beta^nG_\lambda(\tau^n(\bm{\omega}')).
\end{eqnarray*}
Next
\begin{eqnarray*}
|G_\lambda(\bm{\omega})
-G_\lambda(\bm{\omega}')|=
\beta^n
|G_\lambda(\tau^n(\bm{\omega}))
-G_\lambda(\tau^n(\bm{\omega}'))|.
\end{eqnarray*}
Since subsidy sensitivity represents
discounted passive occupation time, $0\le G_\lambda(\bm{\omega})
\le\frac{1}{1-\beta}.$
Thus,
\begin{equation*}
\label{eq:G-passive-n-step-bound}
|G_\lambda(\bm{\omega})
-G_\lambda(\bm{\omega}')|
\le
\frac{\beta^n}{1-\beta}.
\end{equation*}
This estimate requires no prior
Lipschitz assumption.
If both beliefs remain passive forever,
letting $n\to\infty$ gives
\[
G_\lambda(\bm{\omega})
=
G_\lambda(\bm{\omega}')
=
\frac{1}{1-\beta}.
\]
Thus, the subsidy sensitivity is constant
on any set of beliefs whose optimal
trajectories remain passive indefinitely.

\textbf{B) First Active-Action Time}

Consider an optimal policy that selects
the passive action until its first
active-action time.
Define
\[
T_\lambda(\bm{\omega})
=
\inf\{t\ge0:A_t=1\}.
\]
Suppose that
$T_\lambda(\bm{\omega})=n<\infty.$
Since the first $n$ actions are passive,
the accumulated discounted passive
occupation time is
\[
\sum_{k=0}^{n-1}\beta^k
=
\frac{1-\beta^n}{1-\beta}.
\]
At time $n$, the active action reveals
the hidden state, and the subsequent
belief is $\Gamma_s$ with probability
$[\tau^n(\bm{\omega})]_s$.
Assuming differentiability of the
continuation values and occupation-time
consistency of the optimal policy,
we obtain
\begin{equation}
\label{eq:G-first-active}
\begin{aligned}
G_\lambda(\bm{\omega})
&=
\frac{1-\beta^n}{1-\beta}+
\beta^n
\sum_{s=0}^{2}
[\tau^n(\bm{\omega})]_s
G_\lambda(\Gamma_s).
\end{aligned}
\end{equation}
This representation expresses subsidy
sensitivity directly through the first
active-action time.

\textbf{C) Lipschitz Estimate on a
Fixed First-Active-Time Region}

Let $C_n\subseteq\Pi$ be a convex
set on which an optimal policy has the same first active-action time $n$.
For any $\bm{\omega},\bm{\omega}'\in C_n$, then using \ref{eq:G-first-active},  we get 
\begin{align*}
G_\lambda(\bm{\omega})
-G_\lambda(\bm{\omega}')
=
\beta^n
\sum_{s=0}^{2}
[\tau^n(\bm{\omega})
-\tau^n(\bm{\omega}')]_s
G_\lambda(\Gamma_s).
\end{align*}
Using the total-variation inequality,
\begin{align*}
&|G_\lambda(\bm{\omega})
-G_\lambda(\bm{\omega}')|\\
&\quad\le
\frac{\beta^n}{2}
\operatorname{span}
\{G_\lambda(\Gamma_s)\}_{s=0}^{2}
\|\tau^n(\bm{\omega})
-\tau^n(\bm{\omega}')\|_1.
\end{align*}
Since $0\le G_\lambda(\Gamma_s)
\le\frac{1}{1-\beta},$
we have $\operatorname{span}
\{G_\lambda(\Gamma_s)\}_{s=0}^{2}
\le\frac{1}{1-\beta}.$
Therefore,
\begin{align*}
|G_\lambda(\bm{\omega})
-G_\lambda(\bm{\omega})|
\le
\frac{\beta^n}{2(1-\beta)}
\|\tau^n(\bm{\omega})
-\tau^n(\bm{\omega})\|_1.
\end{align*}
By repeated application of the
Dobrushin contraction inequality, $\|\tau^n(\bm{\omega})
-\tau^n(\bm{\omega}')\|_1
\le
c_0^n\|\bm{\omega}-\bm{\omega}'\|_1.$
Hence,
\begin{equation*}
|G_\lambda(\bm{\omega})
-G_\lambda(\bm{\omega}')|
\le
\frac{(\beta c_0)^n}{2(1-\beta)}
\|\bm{\omega}-\bm{\omega}\|_1.
\end{equation*}

Further, for any $n \geq 1,$
\begin{equation*}
\label{eq:G-fixed-exit-Lipschitz}
|G_\lambda(\bm{\omega})
-G_\lambda(\bm{\omega}')|
\le
\frac{\beta }{2(1-\beta)}
\|\bm{\omega}-\bm{\omega}\|_1.
\end{equation*}
Hence $L_G = \frac{\beta}{2 (1-\beta)}$ is valid.

\section*{Proof of Discrepancy  Bound $0 \le D_{\Gamma, \tau} \le 2$} \label{Appendix: Bound}

Define the maximum row discrepancy
$d_{01}:=\max_{i,j\in\{0,1,2\}}
\left\|(p^1_{i,\cdot})^\top-(p^0_{j,\cdot})^\top\right\|_1$
Since $\tau(\bm{\omega})
=(P^0)^\top\bm{\omega}
=\sum_{j=0}^{2}\omega_j(p^0_{j,\cdot})^\top$, and $\bm{\Gamma}_s=(p^1_{s,\cdot})^\top$,

\begin{align}
\|\bm{\Gamma}_s-\tau(\bm{\omega})\|_1
&=
\left\|
\sum_{j=0}^{2}\omega_j
\left[(p^1_{s,\cdot})^\top
-(p^0_{j,\cdot})^\top\right]
\right\|_1\nonumber\\
&\leq
\sum_{j=0}^{2}\omega_j
\left\|(p^1_{s,\cdot})^\top
-(p^0_{j,\cdot})^\top\right\|_1.
\end{align}
Multiplying by \(\omega_s\) and summing over \(s\),
 \begin{align} \sum_s\omega_s \|\Gamma_s-\tau(\bm\omega)\|_1 &\le \sum_{s,j}\omega_s\omega_j \|P^1_{s,\cdot}-P^0_{j,\cdot}\|_1\nonumber\\ &\le d_{01}.
 \end{align} 
Therefore, $D_{\Gamma,\tau} \le d_{01}\le2.$

\section{Proof of Proposition~\ref{Prop:admissible-discount-parameter}: Admissible range of the discount factor: }
\label{Proof of Proposition 2}
The inequality $\beta^2D_{\Gamma,\tau}\leq 2(1-\beta)$
is equivalent to $D\beta^2+2\beta-2\leq 0.$
For $D_{\Gamma,\tau}>0$, the roots of the corresponding
quadratic equation are
\[
\beta_{\pm}
=
\frac{-2\pm\sqrt{4+8D_{\Gamma,\tau}}}{2D_{\Gamma,\tau}}.
\]
Simplifying,
\[
\beta_{\pm}
=
\frac{-1\pm\sqrt{1+2D_{\Gamma,\tau}}}{D_{\Gamma,\tau}}.
\]
Since $D_{\Gamma,\tau}>0$,
\[
\beta_-<0,
\qquad
\beta_+>0.
\]
The quadratic polynomial $f(\beta)=D_{\Gamma,\tau}\beta^2+2\beta-2$
is convex in $\beta$.
Therefore, $f(\beta)\leq 0$
for $\beta\in[\beta_-,\beta_+].$
Restricting to $\beta\in(0,1)$ gives
\[
0<\beta\leq
\frac{\sqrt{1+2D_{\Gamma,\tau}}-1}{D_{\Gamma,\tau}}.
\]
Rationalizing the numerator,
\begin{align*}
\frac{\sqrt{1+2D_{\Gamma,\tau}}-1}{D_{\Gamma,\tau}}
&=
\frac{(1+2D_{\Gamma,\tau})-1}
{D_{\Gamma,\tau}(\sqrt{1+2D_{\Gamma,\tau}}+1)}\\
&=
\frac{2}{1+\sqrt{1+2D_{\Gamma,\tau}}}.
\end{align*}
Thus,
\[
\beta_{\max}(D_{\Gamma,\tau})
=
\frac{2}{1+\sqrt{1+2D_{\Gamma,\tau}}}.
\]
If $D_{\Gamma,\tau}=0$, the original condition reduces to
\[
0\leq 2(1-\beta),
\]
which holds for every $\beta\in(0,1)$.

\section{Joint continuity of the value function}
\begin{lemma}[Convexity and joint continuity of the value function]
\label{lem:value-continuity-convexity}
For every belief $\bm{\omega}\in\Pi$,
\begin{enumerate}
    \item the optimal value function $V^\lambda(\bm{\omega})$
is convex and nondecreasing in $\lambda.$
\item  moreover,
\[
|V^{\lambda_1}(\bm{\omega})
-V^{\lambda_2}(\bm{\omega})|
\le
\frac{|\lambda_1-\lambda_2|}{1-\beta}.
\]
\item  For every fixed $\lambda$,
\[
|V^\lambda(\bm{\omega})
-V^\lambda(\bm{\omega}')|
\le
\frac{r_{\max}+|\lambda|}{1-\beta}
\|\bm{\omega}-\bm{\omega}'\|_1,
\]
where, $|r(s,a)|\le r_{\max}.$
\item  $V^\lambda(\bm{\omega})$ is jointly
continuous in $(\bm{\omega},\lambda).$
\end{enumerate}
\end{lemma}
\begin{proof}
\begin{enumerate}
    \item Proof convexity of value function in $\lambda$ is using the mathematical induction. 
The sketch is as follows. 
The pointwise maximum of affine function is a convex function. For detailed steps see \cite[Appendix B, Proof of Lemma $2$ Part $2$]{Meshram18}.
\item Let $\Phi$ denote the common class of admissible
history-dependent policies, independent of the
initial belief and subsidy parameter.
For each $\phi\in\Phi$, we write
\[
V_\phi^\lambda(\bm{\omega})
=
J_\phi(\bm{\omega})
+\lambda H_\phi(\bm{\omega}),
\]
For any $\pi\in\Pi$, define the expected
discounted reward without subsidy as
\[
J_\phi(\bm{\omega})
:=
\mathbb E_{\phi,\bm{\omega}}
\left[
\sum_{t=0}^{\infty}
\beta^t r(S_t,A_t)
\right],
\]
and the expected discounted passive
occupation time as
\[
H_\phi(\bm{\omega})
:=
\mathbb E_{\phi,\bm{\omega}}
\left[
\sum_{t=0}^{\infty}
\beta^t\mathbf{1}\{A_t=0\}
\right].
\]
Hence $0\le H_{\phi}(\bm{\omega})
\le\frac{1}{1-\beta}.$

 $V_{\phi}^{\lambda}(\bm{\omega})$ is convex in $\lambda$ and there exists $\lambda_L,\lambda_H \in \mathbb{R}$ such that it is convex on $(\lambda_L,\lambda_H).$  $V_{\phi}^{\lambda}(\bm{\omega})$ has right and left derivatives at each $\lambda \in (\lambda_L,\lambda_H).$ Further $V_{\phi}^{\lambda}(\bm{\omega})$ is Lipschitz and therefore absolutely continuous on each closed, bounded subinterval of $(\lambda_L,\lambda_H).$ Also, $V_{\phi}^{\lambda}(\bm{\omega})$ is differentiable except at a countable number of points. These results are well studied for convex function in \cite[Chapter $6,$ Section $6.6$]{Royden2010}.

For any  $\lambda,\lambda'$,
\[
|V_\phi^\lambda(\bm{\omega})
-V_\phi^{\lambda'}(\bm{\omega})|
\le
\frac{|\lambda-\lambda'|}{1-\beta}.
\]
Taking the supremum over policies yields
\begin{equation*}
\label{eq:subsidy-value-lipschitz}
|V^\lambda(\bm{\omega})
-V^{\lambda'}(\bm{\omega})|
\le
\frac{|\lambda-\lambda'|}{1-\beta}.
\end{equation*}
\item For each fixed policy $\phi$, conditioning
on the initial hidden state gives
\[
V_\phi^\lambda(\bm{\omega})
=
\sum_{s\in\mathcal S}
\omega_s v_\phi^\lambda(s),
\]
where $v_\phi^\lambda(s)$ denotes the
discounted value conditional on $S_0=s$.
By bounded rewards and geometric discounting, we have 
\[
|v_\phi^\lambda(s)|
\le
\frac{r_{\max}+|\lambda|}{1-\beta}.
\]
Further, as value function is also convex and Lipschitz in $\bm{\omega}$  we have 
\begin{align*}
|V_\phi^\lambda(\bm{\omega})
-V_\phi^\lambda(\bm{\omega}')|
\le
\frac{r_{\max}+|\lambda|}{1-\beta}
\|\bm{\omega}-\bm{\omega}'\|_1.
\end{align*}

Using
\[
\left|\sup_\phi f_\phi-\sup_\phi g_\phi\right|
\le
\sup_\phi|f_\phi-g_\phi|,
\]
we obtain
\begin{equation*}
\label{eq:belief-value-lipschitz}
|V^\lambda(\bm{\omega})
-V^\lambda(\bm{\omega}')|
\le
\frac{r_{\max}+|\lambda|}{1-\beta}
\|\bm{\omega}-\bm{\omega}'\|_1.
\end{equation*}

\item  Finally, the triangle inequality gives
\begin{align*}
|V^\lambda(\bm{\omega})
-V^{\lambda'}(\bm{\omega}')|
\le
|V^\lambda(\bm{\omega})
-V^{\lambda'}(\bm{\omega})| + \\ 
|V^{\lambda'}(\bm{\omega})
-V^{\lambda'}(\bm{\omega}')|.
\end{align*}
From preceding two steps, $V^\lambda(\bm{\omega})$
is joint Lipschitz and  
joint continuity follows immediately.

\end{enumerate}
\end{proof}

\section{Non-differentiability of Value function in $\lambda$ and Lipschitz constant}

We provide more general results when value function is not differentiable in $\lambda,$ but using subdifferential provide desired Lipschitz constant. 

First, we present background results.

\textbf{Ordinary Lipschitz continuity}
A map $f:\mathbb R^n\to\mathbb R$ is Lipschitz continuous on $X\subseteq\mathbb R^n$ if there exists $L\geq0$ such that
\begin{equation*}\label{eq:ordinary}
 |f(x)-f(y)|\leq L\|x-y\|,\qquad x,y\in X.
\end{equation*}

For convex nonsmooth functions, $\partial f(x)$ is generally a set. Thus a distance between sets is needed to study the continuity of $x\mapsto\partial f(x)$.

\textbf{Hausdorff distance} We provide definition and it will be used subsequently for our analysis. 
\begin{definition}[Distance from a point to a set]
For $b\in\mathbb R^m$ and nonempty $C\subseteq\mathbb R^m$, define
\[
\operatorname{dist}(b,C)=\inf_{z\in C}\|b-z\|.
\]
\end{definition}
\begin{definition}[Hausdorff distance]
For nonempty compact $C,D\subseteq\mathbb R^m$, the Hausdorff distance is
\begin{equation*}
\label{eq:hausdorff}
 d_H(C,D)=\max\left\{\sup_{a\in C}\operatorname{dist}(a,D),\;\sup_{b\in D}\operatorname{dist}(b,C)\right\}.
\end{equation*}
\end{definition}
This measures the greatest distance required to move an element of either set to the other set.

\begin{lemma}[Intervals]
If $C=[a_-,a_+]$ and $D=[b_-,b_+]$ are nonempty compact intervals in $\mathbb R$, then
\begin{equation*}
\label{eq:interval}
 d_H(C,D)=\max\{|a_--b_-|,\;|a_+-b_+|\}.
\end{equation*}
\end{lemma}
\begin{proof}
For intervals, the largest directed deviation of either set from the other is attained at an endpoint. Taking the maximum of the two directed deviations yields desired result.
\end{proof}

\textbf{Hausdorff Lipschitz continuity} This is defined for set valued function and it is with respect to Hausdorff distance.
\begin{definition}
Let $F:X\rightarrow\mathbb R^m$ have nonempty compact values. The multifunction $F$ is Hausdorff-Lipschitz continuous on $X$ if some $L\geq0$ satisfies
\begin{equation}\label{eq:HL}
 d_H(F(x),F(y))\leq L\|x-y\|,\qquad x,y\in X.
\end{equation}
\end{definition}
The property requires the entire set $F(x)$ to vary at a controlled rate, not merely one selection from it.

\textbf{Subgradients and the convex sum rule}
For a proper convex function $f:\mathbb R^n\to(-\infty,+\infty]$, a vector $u$ belongs to $\partial f(x)$ precisely when
\[
 f(y)\ge f(x)+\langle u,y-x\rangle\quad\text{for every }y.
\]
In particular, $0\in\partial f(x)$ if and only if $x$ minimizes $f$ globally. 
If $f$ and $g$ are proper convex functions, then
\[
\partial f(x)+\partial g(x)\subseteq\partial(f+g)(x).
\]
\begin{proof}
For $u\in\partial f(x)$ and $v\in\partial g(x)$, add the two supporting-hyperplane inequalities to obtain
$(f+g)(y)\ge(f+g)(x)+\langle u+v,y-x\rangle$.
\end{proof}
Equality holds, for example, when $f$ and $g$ are proper lower-semicontinuous convex functions with intersecting relative interiors of their domains; in particular, it holds for finite convex functions on all of $\mathbb R^n.$


We have 
\begin{equation*}
\label{eq:sup}
V_\lambda(\bm\omega)=\sup_{\phi\in\Phi}\bigl\{J_{\phi}(\bm\omega)+\lambda H_{\phi}(\bm\omega)\bigr\}.
\end{equation*}
For compact subsets $C,D\subset\mathbb R$, their Hausdorff distance is
\begin{equation}
d_H(C,D)=\max\left\{\sup_{c\in C}\inf_{d\in D}|c-d|,\ \sup_{d\in D}\inf_{c\in C}|d-c|\right\}.
\end{equation}
In particular, $d_H([a,b],[c,d])=\max\{|a-c|,|b-d|\}$.

\textbf{Subdifferential representation} We now obtain the bound on subdifferential. 
\begin{proposition}[Universal subsidy sensitivity]\label{prop:basic}
For each $\bm\omega$, $\lambda\mapsto V_\lambda(\bm\omega)$ is finite, convex, nondecreasing, and $1/(1-\beta)$ Lipschitz. Moreover,
\begin{equation*}
\label{eq:basic}
\partial_\lambda V_\lambda(\bm\omega)\subseteq\left[0,\frac1{1-\beta}\right].
\end{equation*}
\end{proposition}
\begin{proof}
Each summand in value function is affine in $\lambda$ with slope $H_\phi(\bm\omega)\in[0,(1-\beta)^{-1}]$. The supremum is convex and nondecreasing. For any $\lambda,\lambda'$,
\[
|V_\lambda(\bm\omega)-V_{\lambda'}(\bm\omega)|\le\sup_\phi H^\phi(\bm\omega)|\lambda-\lambda'|\le\frac{|\lambda-\lambda'|}{1-\beta}.
\]
Every convex subgradient lies between the left and right derivatives, which are in the same slope interval. Hence $\partial_\lambda V_\lambda(\bm\omega)\subseteq[0,\frac{1}{1-\beta}]$ \cite[Section 6.6]{Royden2010}.
\end{proof}

\begin{assum}
\label{ass:attain}
For each $(\bm\omega,\lambda)$ of interest, the supremum in value function is attained and the envelope representation holds:
\begin{equation}\label{eq:envelope}
\partial_\lambda V_\lambda(\bm\omega)=\mathrm{conv} \{H_{\phi}(\bm\omega):\phi\in\Phi_\lambda^\star(\bm\omega)\},
\end{equation}
where $\Phi_\lambda^\star(\bm\omega)$ is the set of maximizing policies, and the set on the right is closed. This holds,  if the relevant attainable reward-occupation pairs form a compact set. 
\end{assum}
By Assumption we have,
\begin{equation}\label{eq:endpoints}
\partial_\lambda V_\lambda(\bm\omega)=[H^-_\lambda(\bm\omega),H^+_\lambda(\bm\omega)],
\end{equation}
where $H^-_\lambda=\min_{\phi\in\Phi_\lambda^\star}H_\phi$ and $H^+_\lambda=\max_{\phi\in\Phi_\lambda^\star}H_\phi.$

\begin{proposition}
For finite convex function $V_\lambda(\bm\omega)$, we have 
\begin{equation}\label{eq:subinterval}
 \partial_\lambda V_\lambda(\bm\omega)
 =[H^-_\lambda(\bm\omega),H^+_\lambda(\bm\omega)].
\end{equation}
Consequently,
\begin{align}\label{eq:endpointdistance}
 &d_H\bigl(\partial_\lambda V_\lambda(\bm\omega),
 \partial_\lambda V_\lambda(\bm\omega')\bigr)\\
 &\quad=\max\bigl\{|H^-_\lambda(\bm\omega)-H^-_\lambda(\bm\omega')|,
 |H^+_\lambda(\bm\omega)-H^+_\lambda(\bm\omega')|\bigr\}.\notag
\end{align}
\end{proposition}
\begin{proof}
The subdifferential of a finite one-dimensional convex function is the interval between its left and right derivatives \cite[Theorem 27.4]{munkres2000topology}. Apply Lemma 2.3, the interval formula \eqref{eq:interval}, to these two intervals.
\end{proof}

\textbf{Four deterministic switching policies}
Let $n\ge1$ be a deterministic switching time, independent of the initial state and observation history. Direct summation gives Table~\ref{tab:policies}.
\begin{table}[ht]
\centering\small
\caption{Discounted passive occupation for deterministic schedules}\label{tab:policies}
\begin{tabular}{@{}p{0.29\linewidth}p{0.32\linewidth}p{0.28\linewidth}@{}}
\toprule
Policy & Actions & $H_\phi(\bm{\omega})$\\\midrule
Always passive & $A_t=0$ for all $t$ & $1/(1-\beta)$\\
Always active & $A_t=1$ for all $t$ & $0$\\
Passive then active & $A_t=0$ for $t<n$; $A_t=1$ for $t\ge n$ & $(1-\beta^n)/(1-\beta)$\\
Active then passive & $A_t=1$ for $t<n$; $A_t=0$ for $t\ge n$ & $\beta^n/(1-\beta)$\\
\bottomrule
\end{tabular}
\end{table}
For any such fixed policy,
\begin{equation}
\partial_\lambda V^\pi_\lambda(\bm{\omega})=\{H_\phi\},\quad d_H\bigl(\partial_\lambda V^\pi_\lambda(\bm{\omega}),\partial_\lambda V^\pi_\lambda(\bm{\omega}')\bigr)=0.
\end{equation}

\textbf{Fixed-policy Lipschitz bound}
\begin{lemma}\label{lem:fixed}
Let $\phi$ be a fixed history-dependent policy whose time-zero action is the same deterministic action $a_0\in\{0,1\}$ under every initial state. Then, for all $\bm{\omega},\bm{\omega}'\in\Pi$,
\begin{equation}\label{eq:fixed}
|H_\phi(\bm{\omega})-H_\phi(\bm{\omega}')|\le L_G\|\bm{\omega}-\bm{\omega}'\|_1,\qquad L_G:=\frac{\beta}{2(1-\beta)}.
\end{equation}
\end{lemma}
\begin{proof}
Conditioning on the initial hidden state $S_0=i$, define
\[
h_{\phi,i}=\mathbb{E}_{\phi,i}\!\left[\sum_{t=1}^\infty\beta^{t-1}\mathbf1_{\{A_t=0\}}\right],\quad 0\le h_{\phi,i}\le\frac1{1-\beta}.
\]
Because the policy is the same history-based rule, total expectation yields
\[
H_\phi(\bm{\omega})=\mathbf1_{\{a_0=0\}}+\beta\sum_i\bm{\omega}^i h_{\phi,i}.
\]
Set $d_i=\bm\omega^i-\bm\omega^{i,'},$ then $\sum_i d_i=0$.
Since both beliefs are probability vectors,
\[
\sum_{i=0}^{2}d_i
=
\sum_{i=0}^{2}\bm\omega^i
-
\sum_{i=0}^{2}\bm\omega^{i,'}
=0.
\]

For any real vector $h$, define $h_{\max}=\max_i h_i,$  $h_{\min}=\min_i h_i,$
and $\operatorname{span}(h)=h_{\max}-h_{\min}.$
Choose the midpoint $\hat{c}=\frac{h_{\max}+h_{\min}}{2}.$ 
Since $\sum_i d_i=0$, we have
\begin{align*}
\sum_i d_i h_i
&=\sum_i d_i h_i- \hat{c}\sum_i d_i\\
&=\sum_i d_i(h_i-\hat{c}).
\end{align*}

Applying the triangle inequality gives
\begin{align*}
\left|\sum_i d_i h_i\right|
&=\left|\sum_i d_i(h_i-\hat{c})\right|\\
&\le \sum_i |d_i|\,|h_i-\hat{c}|.
\end{align*}

Since $h_{\min}\le h_i\le h_{\max}$,
\[
|h_i-\hat{c}|
\le
\frac{h_{\max}-h_{\min}}{2}
=
\frac{\operatorname{span}(h)}{2}.
\]
Then 
\begin{align*}
\left|\sum_i d_i h_i\right|
&\le
\sum_i |d_i|
\frac{\operatorname{span}(h)}{2}\\
&=
\frac{\operatorname{span}(h)}{2}
\sum_i |d_i|\\
&=
\frac{\operatorname{span}(h)}{2}
\|\bm{\omega}-\bm{\omega}'\|_1.
\end{align*}

Thus
\begin{align*}
|H_\phi(\bm\omega)-H_\phi(\bm\omega')|\le\frac{\beta}{2}\mathrm{span}(h_\phi)\|\bm\omega-\bm\omega'\|_1 \\
\le
\frac{\beta}{2(1-\beta)}\|\bm\omega-\bm\omega'\|_1.
\end{align*}

\end{proof}

\textbf{Hausdorff-Lipschitz theorem for the optimal subdifferential}
\begin{assum}[Common optimal policy family]\label{ass:stable}
Fix a subsidy $\lambda$ and a subset $U\subseteq\Pi$. There exists a nonempty collection $\Phi^\star$ of \emph{fixed history-based policy rules} such that $\Phi_\lambda^\star(\bm\omega)=\Phi^\star$ for every $\bm\omega\in U$. Each $\phi\in\Phi^\star$ has a deterministic initial action independent of the initial state. The initial action may differ between policies.
\end{assum}
\begin{theorem}[Hausdorff-Lipschitz continuity]\label{thm:main}
Under Assumptions~\ref{ass:attain} and \ref{ass:stable}, for every $\bm\omega,\bm\omega'\in U$,
\begin{equation}\label{eq:main}
d_H\!\left(\partial_\lambda V_\lambda(\bm\omega),\partial_\lambda V_\lambda(\bm\omega')\right)\le\frac{\beta}{2(1-\beta)}\|\bm\omega-\bm\omega'\|_1.
\end{equation}
\end{theorem}
\begin{proof}
By the envelope representation,
\[
H^-_\lambda(\omega)=\inf_{\phi\in\Phi^\star}H_\phi(\bm\omega),\qquad H^+_\lambda(\bm\omega)=\sup_{\phi\in\Phi^\star}H_\phi(\bm\omega).
\]
For two bounded indexed families $(a_\phi)$ and $(b_\phi)$ over the same index set,
\[
|\inf_\phi a_\phi-\inf_\phi b_\phi|\le\sup_\phi|a_\phi-b_\phi|,
\]
with the analogous inequality for suprema.

Lemma~\ref{lem:fixed} applies uniformly to every $\phi\in\Phi^\star$, so
\begin{align*}
|H^-_\lambda(\bm\omega)-H^-_\lambda(\bm\omega')|&\le L_G\|\bm\omega-\bm\omega'\|_1,\\
|H^+_\lambda(\omega)-H^+_\lambda(\omega')|&\le L_G\|\omega-\omega'\|_1.
\end{align*}

By the envelope representation,
\begin{align*}
\partial_\lambda V_\lambda(\bm\omega)
&=
[H^-_\lambda(\bm\omega),
 H^+_\lambda(\bm\omega)],\\
\partial_\lambda V_\lambda(\bm\omega')
&=
[H^-_\lambda(\bm\omega'),
 H^+_\lambda(\bm\omega')].
\end{align*}

Let $A=[a,b]$ and $B=[c,d]$ be two
nonempty compact intervals in $\mathbb R$.
By definition,
\begin{align*}
d_H(A,B)
=\max\Big\{
\sup_{x\in A}\operatorname{dist}(x,B),\sup_{y\in B}\operatorname{dist}(y,A)
\Big\},
\end{align*}
where
\[
\operatorname{dist}(x,B)
=\inf_{y\in B}|x-y|.
\]

For $x\in[a,b]$, we have
\[
\operatorname{dist}(x,B)
=
\max\{c-x,0,x-d\}.
\]
Consequently,
\[
\sup_{x\in A}\operatorname{dist}(x,B)
=
\max\{(c-a)_+,(b-d)_+\},
\]
where $(z)_+=\max\{z,0\}$.

Similarly,
\[
\sup_{y\in B}\operatorname{dist}(y,A)
=
\max\{(a-c)_+,(d-b)_+\}.
\]

It follows that
\begin{align*}
d_H(A,B)
&=
\max\{(c-a)_+,(a-c)_+,\\
&\qquad (b-d)_+,(d-b)_+\}\\
&=\max\{|a-c|,|b-d|\}.
\end{align*}

Applying this identity to the subdifferential
intervals gives
\begin{align*}
&d_H\bigl(
\partial_\lambda V_\lambda(\bm\omega),
\partial_\lambda V_\lambda(\bm\omega')
\bigr)\\
&=
\max\Big\{
|H^-_\lambda(\bm\omega)
-H^-_\lambda(\bm\omega')|,\\
&\qquad
|H^+_\lambda(\bm\omega)
-H^+_\lambda(\bm\omega')|
\Big\}.
\end{align*}

By Step 3, each endpoint difference is bounded
above by
\[
L_G\|\bm\omega-\bm\omega'\|_1,
\qquad
L_G=\frac{\beta}{2(1-\beta)}.
\]

Therefore,
\begin{align*}
&d_H\bigl(
\partial_\lambda V_\lambda(\bm\omega),
\partial_\lambda V_\lambda(\bm\omega')
\bigr)\\
&\le
\max\Big\{
L_G\|\bm\omega-\bm\omega'\|_1,\\
&\qquad
L_G\|\bm\omega-\bm\omega'\|_1
\Big\}\\
&=
L_G\|\bm\omega-\bm\omega'\|_1\\
&=
\frac{\beta}{2(1-\beta)}
\|\bm\omega-\bm\omega'\|_1.
\end{align*}

This establishes \eqref{eq:main}
and completes the proof.

\end{proof}


\begin{remark}
Assumption~\ref{ass:stable} is a sufficient policy-stability condition, not a primitive condition on rewards and transition matrices. The theorem is local to $U$ and is not automatically global on the simplex. Policy changes at boundaries require a separate compatibility argument.
\end{remark}


\begin{proposition}\label{prop:boundary}
Suppose two distinct fixed switching policies are uniquely optimal on opposite sides of a boundary point $\bm\omega^\star$, and their occupation measures have continuous limits there. If the optimal-value subdifferential is Hausdorff-continuous at $\bm\omega^\star$, then the two limiting occupation measures must agree.
\end{proposition}
\begin{proof}
Approach $\bm\omega^\star$ from either side. The subdifferentials at uniquely optimal points are singletons. Hausdorff continuity forces both singleton limits to equal the subdifferential at $\bm\omega^\star$, so their elements must coincide.
\end{proof}

\textbf{Conditional implication for a collapsing bandit}
 Write the passive-minus-active action-value difference as
\begin{equation}\label{eq:delta}
\Delta(\bm\omega,\lambda)=d(\bm\omega)+\lambda+\beta V_\lambda(\tau(\bm\omega))-\beta\sum_{i=0}^{2}\bm\omega^i V_\lambda(\Gamma_i).
\end{equation}
The function $\Delta$ is a difference of convex functions. Thus a convex subdifferential sum rule must \emph{not} be applied to \eqref{eq:delta}.

For a finite convex function, denote the left and right derivatives by $V'_{\lambda,-}$ and $V'_{\lambda,+}$. By \eqref{eq:endpoints}, these are $H^-_\lambda$ and $H^+_\lambda$, respectively. Taking one-sided derivatives in \eqref{eq:delta} gives
\begin{equation}\label{eq:derivative}
\Delta'_{\pm}(\bm\omega,\lambda)=1+\beta H^\pm_\lambda(\tau(\bm\omega))-\beta\sum_i\bm\omega^i H^\pm_\lambda(\Gamma_i).
\end{equation}

\begin{corollary}\label{cor:mono}
Assume the Hausdorff bound \eqref{eq:main} holds, with the same constant $L_G$, between $\tau(\bm\omega)$ and each $\Gamma_i$ for every subsidy in an interval $I$. Then
\begin{equation}\label{eq:lower}
\Delta'_{\pm}(\bm\omega,\lambda)\ge1-\beta L_G\sum_i\bm\omega_i\|\tau(\bm\omega)-\Gamma_i\|_1.
\end{equation}
In particular, if
\begin{equation}\label{eq:condition}
\frac{\beta^2}{2(1-\beta)}\sup_{
\bm\omega\in\Pi}\sum_i\bm\omega^i\|\tau(\bm\omega)-\Gamma_i\|_1<1,
\end{equation}
and the assumed Hausdorff bound applies to every pair of successor beliefs under consideration, then $\Delta(\bm\omega,\lambda)$ is strictly increasing in $\lambda\in I$ for every $\bm\omega$.
\end{corollary}
\begin{proof}
For either sign, write the difference of sensitivities in \eqref{eq:derivative} as
\[
H^\pm_\lambda(\tau(\bm\omega))-\sum_i\bm\omega^iH^\pm_\lambda(\Gamma_i)
=\sum_i\bm\omega_i\bigl[H^\pm_\lambda(\tau(\bm\omega))-H^\pm_\lambda(\Gamma_i)\bigr].
\]
The Hausdorff bound between intervals bounds each endpoint difference by $L_G\|\tau(\omega)-\Gamma_i\|_1$. This yields \eqref{eq:lower}. Under \eqref{eq:condition}, both one-sided derivatives have a uniform positive lower bound. Since $\Delta$ is locally Lipschitz in $\lambda$, it is absolutely continuous on compact intervals \cite[Section 6.4]{Royden2010}; integrating its almost-everywhere derivative proves strict increase.
\end{proof}

\begin{remark}
Strict increase of $\Delta$ in the subsidy provides the single-crossing direction for passive optimality. 
\end{remark}

\subsection{Various bounds on $D_{\Gamma, \tau}$ under different model assumptions}

We next investigate the sufficient condition on $D_{\Gamma,\tau}$ and illustrate the trade-off between the discount parameter and the discrepancy parameter $D_{\Gamma,\tau}.$ We note that the sufficiency condition for indexability is $\beta L_G D_{\Gamma,\tau}\leq 1.$  Furhter for $L_G = \frac{\beta}{2 (1-\beta)}$ we obtain $\beta^2D_{\Gamma,\tau}
\leq 2(1-\beta).$
This inequality characterizes a trade-off between the
discount factor and the discrepancy between the
posterior and predicted beliefs.
A smaller $D_{\Gamma,\tau}$ implies a larger discount
factor, whereas a larger discount factor requires a
smaller posterior-prediction discrepancy.
Since $\Gamma_s$ and $\tau(\bm{\omega})$ are
probability vectors, $\left\|
\Gamma_s-\tau(\bm{\omega})
\right\|_1
\leq 2.$ Hence $0\leq D_{\Gamma,\tau}\leq 2.$








\textbf{Universal Discount-Factor Threshold}

Since $D_{\Gamma,\tau}\leq 2,$ 
a sufficient condition valid for every
three-state transition model is $2\beta^2\leq 2(1-\beta).$ This is equivalent to solving 
\[
\beta^2+\beta-1\leq 0.
\]

The positive root is $\beta_*=\frac{\sqrt{5}-1}{2}.$
Therefore, $0<\beta\leq
\frac{\sqrt{5}-1}{2}
\approx 0.6180$
guarantees that $\beta^2D_{\Gamma,\tau}
\leq 2(1-\beta)$ for every  $D_{\Gamma,\tau}\in[0,2].$ Thus $\beta_*=\frac{\sqrt{5}-1}{2}$ is independent of the transition probability matrices. 
Hence, for $\beta\leq\beta_*$,
the numerical sufficient condition  holds
without imposing additional restrictions
on the transition matrices.

\textbf{Maximum allowable posterior-prediction discrepancy $D_{\Gamma,\tau}$}
For a fixed discount factor $\beta$,
the sufficient condition 
can be written as follows
\begin{equation*}
\label{eq:D-max-beta}
D_{\Gamma,\tau}
\leq
D_{\max}(\beta)
:=
\frac{2(1-\beta)}{\beta^2}.
\end{equation*}
This provides an explicit restriction on
the transition matrices for any chosen
discount factor.

\textbf{Special case: Identical transition matrices}

Suppose $P^0=P^1=P.$
Then $c_0=c_1=\delta(P).$
Therefore, $D_{\Gamma,\tau}
\leq
\frac43\delta(P).$
A sufficient condition  is
\[
\beta^2\frac43\delta(P)
\leq2(1-\beta).
\]
Rearranging, $\delta(P)
\leq
\frac{3(1-\beta)}{2\beta^2}.$

\textbf{Asymptotic interpretation as $\beta\to1$}

The sufficient condition requires $D_{\Gamma,\tau}
\leq
\frac{2(1-\beta)}{\beta^2}.$
As $\beta\to1$, $\frac{1}{\beta^2}\to1.$
Therefore, $\frac{2(1-\beta)}{\beta^2}
\sim
2(1-\beta),$ as $\beta\to1.$
This implies $D_{\Gamma,\tau}
=
O(1-\beta).$
is sufficient to satisfy the required
scaling, provided its leading constant
is small enough.
This means that the posterior-prediction discrepancy must decrease approximately linearly with $1-\beta$ as the discount
factor approaches one.



\textbf{For long-horizon}
For $\beta=0.95$, the condition requires $D_{\Gamma,\tau}\leq0.1108.$
This may be satisfied when the transition
matrices are strongly mixing and their
corresponding rows are sufficiently close.

For $\beta=0.99$,  the condition is $D_{\Gamma,\tau}\leq0.0204.$
Thus, for  decision makers with long horizon length (patients decision maker),
the sufficient condition becomes more restrictive.

\end{document}